\RequirePackage{fix-cm}
\documentclass[smallextended]{svjour3}       
\smartqed  
\usepackage{graphicx}
\usepackage{fancyhdr}
\usepackage{amsmath}
\usepackage[hidelinks]{hyperref} 
\usepackage{array}
\usepackage{tabularx,ragged2e,booktabs}

\usepackage{amsthm,amsfonts,dsfont,mathrsfs}
\usepackage{bbm}
\usepackage{setspace}
\usepackage{enumitem}
\usepackage{algorithm}
\usepackage[noend]{algpseudocode}
\usepackage{caption}
\usepackage{footnote}
\usepackage{amssymb}
\usepackage{amsfonts}
\newcommand{\formatmode}[2]{#2}

\usepackage{comment}

\usepackage{mathtools} 
\usepackage{thmtools, thm-restate}

\renewcommand{\leq}{\leqslant} 
\renewcommand{\geq}{\geqslant}
\newcommand{\ra}{\rangle}
\newcommand{\la}{\langle} 

\newcommand{\set}[1]{\left\{#1\right\}}
\usepackage{bm, color, xcolor}
\newcommand{\B}[1]{\bm{#1}}
\usepackage[normalem]{ulem}
\usepackage{cite}

\definecolor{gainsboro}{rgb}{0.75, 0.75, 0.75}
\definecolor{yifancolor}{rgb}{0.8, 0.1, 0.80}
\definecolor{yifancommentcolor}{rgb}{0.2, 0.1, 0.80}

\newcommand{\fall}{\forall~}

\let\sset = \subseteq

\newcommand{\footremember}[2]{%
   \footnote{#2}%
   \newcounter{#1}%
   \setcounter{#1}{\value{footnote}}
}
\newcommand{\footrecall}[1]{%
   \footnotemark[\value{#1}]
}

\theoremstyle{definition}
\newtheorem{defn}[theorem]{Definition}

  \let\gc=\gamma  \let\gee=\varepsilon
            
    \let\gs=\sigma

\newcommand{\bA}{\B{A}}
\newcommand{\bB}{\B{B}}
\newcommand{\bC}{\B{C}}
\newcommand{\ba}{\B{a}}
\newcommand{\bb}{\B{b}}
\newcommand{\bu}{\B{u}}
\newcommand{\bv}{\B{v}}

\newcommand{\bI}{\B{I}}
\newcommand{\bP}{\B{P}}
\newcommand{\cI}{\mathcal{I}}

\DeclareMathOperator{\R}{\mathbb{R}}

\DeclareMathOperator{\C}{\mathbb{C}}
\DeclareMathOperator{\cR}{\mathcal{R}}
\DeclareMathOperator{\cL}{\mathcal{L}}

\DeclareMathOperator{\spp}{span}
\DeclareMathOperator*{\argmin}{argmin}
\DeclareMathOperator{\rank}{rank}

\newcommand{\ab}[1]{\la #1 \ra}
\newcommand{\csg}[2]{{#1}_{[#2,\cdot]}}
\newcommand{\rsg}[2]{{#1}_{[\cdot, #2]}}

\DeclareMathOperator{\spcexp}{GridExp}
\DeclareMathOperator{\vspcexp}{VExpansion}
\DeclareMathOperator{\hspcexp}{HExpansion}
\DeclareMathOperator{\vcl}{VCollapse}

\DeclareMathOperator{\cdg}{CDG}
\DeclareMathOperator*{\merge}{Merge}

\newcommand{\edits}[1]{#1}
\newcommand{\editstwo}[1]{#1}

\spnewtheorem{corollary}[theorem]{Corollary}{\bfseries}{\itshape}
\spnewtheorem{lemma}[theorem]{Lemma}{\bfseries}{\itshape}
\spnewtheorem{proposition}[theorem]{Proposition}{\bfseries}{\itshape}

\begin{document}

\formatmode{
\title{Communication Lower Bounds for Nested Bilinear Algorithms}
\author{Caleb Ju\footremember{eq}{Equal contribution. Work was partially done while at the University of Illinois at Urbana-Champaign.}\footremember{gt}{School of Industrial and Systems Engineering, Georgia Institute of Technology \texttt{cju33@gatech.edu}} \and Yifan Zhang\footrecall{eq}\footremember{ut}{Department  of Mathematics  and Oden Institute  for Computational  Engineering and Sciences, University of Texas at Austin, \texttt{yifanz@utexas.edu}} \and Edgar Solomonik\footremember{uiuc}{Department of Computer Science, University of Illinois at Urbana-Champaign, \texttt{solomon2@illinois.edu}} } 
}{
\title{Communication lower bounds for nested bilinear algorithms via rank expansion of Kronecker products\thanks{C.J. and Y.Z. had equal contribution. Work was partially done while at the University of Illinois at Urbana-Champaign.}}



\author{Caleb Ju         \and
        Yifan Zhang      \and
        Edgar Solomonik
}


\institute{Caleb Ju \at
              Industrial and Systems Engineering, Georgia Institute of Technology, Atlanta, GA, 30332, USA. \\
              \email{cju33@gatech.edu}           
           \and
           Yifan Zhang \at
              Oden Institute for Computational Engineering and Sciences, University of Texas at Austin, TX, 78712, USA.\\
              \email{yf.zhang@utexas.edu}
           \and
           Edgar Solomonik \at
              Department of Computer Science, University of Illinois at Urbana-Champaign, IL, 61801, USA. \\
           Corresponding author \\
           Phone: 217-300-4794 \\
           \email{solomon2@illinois.edu}
}

\date{Received: date / Accepted: date}
}

\maketitle

\begin{abstract}
We develop lower bounds on communication in the memory hierarchy or between processors for nested bilinear algorithms, such as 
Strassen's algorithm for matrix multiplication.
We build on a previous framework that establishes communication lower bounds by use of the rank expansion, or the minimum rank of any fixed size subset of columns of a matrix, for each of the three matrices encoding a bilinear algorithm. 
\editstwo{This framework provides lower bounds for a class of dependency directed acyclic graphs (DAGs) corresponding to the execution of a given bilinear algorithm, in contrast to other approaches that yield bounds for specific DAGs. However, our lower bounds only apply to executions that do not compute the same DAG node multiple times.}
Two bilinear algorithms can be nested by taking Kronecker products between their encoding matrices.
Our main result is a lower bound on the rank expansion of a matrix constructed by a Kronecker product derived from lower bounds on the rank expansion of the Kronecker product's operands.
We apply the rank expansion lower bounds to obtain novel communication lower bounds for nested Toom-Cook convolution, Strassen's algorithm, and fast algorithms for contraction of partially symmetric tensors.

\end{abstract}

\keywords{Communication lower bounds \and Bilinear algorithm \and Kronecker product \and Rank expansion \and Strassen's algorithm \and Convolution \and Tensor contraction} 

\subclass{15A03 \and 65F99 \and 65Y05 \and 68Q11}
 
\section{Introduction} \label{sec:intro}
In high-performance computing, communication cost (i.e., data movement across
the memory hierarchy and/or between processors) has been shown to be more
time-consuming and energy-draining than arithmetic
costs~\cite{kogge2013exascale}. And the gap is expected to continue to grow.
\editstwo{The increased utilization of hierarchical memory and mesh
topology (to address physical constraints in hardware, e.g.,
speed-of-light and voltage limitations) favors methods exploiting data
locality, further motivating well-designed communication
patterns~\cite{bilardi1995horizons,bilardi1999processor}.} Therefore, it is
imperative to design algorithms that minimize communication.  Communication
lower bounds provide a theoretical limit and guide the design of algorithms
that minimize
communication~\cite{jia1981complexity,ballard2011minimizing,dinh2020communication}.

\editstwo{The pioneering work by Yao~\cite{yao1979some} introduced 
communication lower bounds for computing a boolean function
between processors by viewing the computation as a decision tree}. 
Later, Hong and Kung introduced the study of communication lower bounds 
for several algorithms {in the memory hierarchy} by modeling the
computation as a \edits{dependency directed acyclic graph, or dependency DAG},
and representing the data access patterns through a red-blue pebble
game~\cite{jia1981complexity}.  Since then, new techniques have been developed
to derive more lower bounds.  \edits{For nested loop programs with relatively
simple access patterns, such as the classical matrix multiplication and LU
factorization, volumetric inequalities such as the Loomis-Whitney
inequality~\cite{loomis1949inequality} or the more general
H\"{o}lder-Brascamp-Lieb inequalities~\cite{holder1889uber,brascamp1976best}
can be used to derive communication lower
bounds~\cite{irony2004communication,demmel2018communication,ballard2013communication,ballard2018communication,ballard2011minimizing}.}
However, fast algorithms such as Strassen's subcubic \editstwo{matrix multiplication}
algorithm~\cite{strassen1969gaussian} have complicated data access patterns, so
volumetric inequalities do not solely suffice in this setting.

\edits{Another approach is to directly study the dependency DAG. 
This approach applies graph-theoretic tools, such
as the graph expansion, dominating sets, or graph partitioning by
eigenvalues~\cite{jia1981complexity,bilardi2019complexity,ballard2013graph,jain2020spectral,de2019complexity,bilardi2017complexity,bilardi1999processor}.}
In general, this family of techniques aims to uncover bounds on the size of
minimum cuts or vertex separators of the dependency DAG. To derive closed-form
communication lower bounds, one needs additional constraints or a priori
information on the dependency DAG. For example, the graph expansion
argument~\cite{ballard2013graph} requires the DAG to have bounded degrees to
derive sharp lower bounds, and the graph partitioning
argument~\cite{jain2020spectral} requires one to know the eigenvalues of the
corresponding Laplacian matrix to derive analytic lower bounds. Furthermore,
the graph-theoretic techniques must fix a dependency DAG.  However, most
algorithms admit algebraic reorganizations (i.e., computation of different
partial sums) that change the dependency graph and may be more communication
efficient in a particular setting. 

By working with more abstract algorithm representations, a larger space of 
 admissible dependency graphs can be considered simultaneously.
Hypergraphs have been used to capture potential orderings of partial sums~\cite{ballard2016hypergraph}, while bilinear algorithms~\cite{pan1984can} provide a more powerful abstraction for problems that can be posed as bilinear maps on two input sets.
Many important numerical problems fall under this category, including matrix multiplication, convolution, and symmetric tensor contractions, and all known fast algorithms for these problems can be expressed as bilinear algorithms.

A bilinear algorithm $(\bA, \bB, \bC)$ with $\B A \in \mathbb{C}^{m_A \times R}$, $\B B \in \mathbb{C}^{m_B \times R}$, and $\B C \in \mathbb{C}^{m_C \times R}$ computes $\B f(\B x, \B y)=\B C [(\B A^T \B x) \odot (\B B^T \B y)]$, where $\odot$ is the Hadamard product (elementwise or bilinear). The value $R$ is called the \edits{\textit{rank}} of the bilinear algorithm.
When a subset of columns from $\B{A}$, $\B{B}$, or $\B{C}$ is
a low-rank matrix, then the communication costs can be reduced
for executing this portion of the computation. To see why,
let $\B{P}$ consist of a subset of $k$ different columns from
an identity matrix of dimension $R$ so that a portion of the
bilinear algorithm associated with $k$ of the $R$ bilinear products is
$\B{CP} \big[\big((\B{AP})^T \B x\big) \odot \big(\B{BP})^T \B y\big)\big]$. 
We see that $\rank(\B{AP})$, $\rank(\B{BP})$, and $\rank(\B{CP})$ bound the minimum number of linear combinations of inputs needed from $\B x$, $\B y$, and the amount of output information produced, respectively, 
in computing this portion of the bilinear algorithm. 
Lower bounds on the growth of this rank with $k$, i.e., the rank expansion, yield lower bounds on communication for any execution DAG of the bilinear algorithm. Such an execution DAG must only compute linear combinations of inputs or the bilinear forms, \editstwo{while the order of operations and intermediate values can be arbitrarily specified. See~\cite[Definition 4.1]{solomonik2017communication} for more details. In particular, we do not restrict the order in which products and sums are carried out nor prohibit intermediate values from being reused between different outputs, such as in the \textit{independent evaluation} model studied by Hong and Kung~\cite{jia1981complexity}.}
The rank expansion of a matrix also characterizes its Kruskal rank~\cite{kruskal1977three}, which is the smallest $k$ for which any $k$ columns of $\B A$ are linearly independent, i.e., $\rank(\B{AP})=k$ for any choice of $\B P$. We formally define bilinear algorithms and their motivation for developing fast bilinear algorithms in Section~\ref{sec:RkExp}.

We focus on
nested bilinear algorithms~\cite{pan1984can}, which are bilinear algorithm
constructed via Kronecker products of matrices encoding the two
factor bilinear algorithms: $(\B A_1 \otimes \B A_2, \B B_1 \otimes \B B_2, \B C_1 \otimes \B C_2)$.
This abstraction captures 
both recursive and higher-order methods for matrix multiplication, polynomial multiplication, convolution, tensor contractions,
as well as other algorithms.
We show in general the rank expansion for the matrices defining a nested bilinear algorithm is based on the rank expansion of its factors. 
We prove that for a certain class of rank expansion lower bounds $\gs_A$ and $\gs_B$ for $\B{A}$ and $\B{B}$, respectively, 
there exists a rank expansion lower bound $\sigma_C$ for $\B C = \B A \otimes \B B$  
given by \(\sigma_C(k) = \min_{k_A \in [1,n_A], k_B \in [1,n_B], k_Ak_B = k} \sigma_A(k_A)\sigma_B(k_B)\), where 
$n_X = \#\text{cols}(\B{X})$.
This result is a generalization of the identity, $\rank(\B A\otimes \B B) = \rank(\B A)\rank(\B B)$. A formal overview of our results is described in Section~\ref{sec:main_results}.

\edits{
To prove our result, we start by introducing the grid framework in Section~\ref{sec:NestedSigma}. This framework provides a visual interpretation for our matrix rank analysis.
Using the grid framework, we show
it suffices to consider the rank of a subset of columns of $\B{A} \otimes \B{B}$ that has a compact geometric structure on a 2D grid. With this structure, we prove two main theorems in Section~\ref{sec:cont_step}. First, we show how to use the compact geometric structure to lower bound the rank of a matrix by solving a nontrivial discrete optimization problem. Second, to simplify the optimization problem, we apply a continuous relaxation. We can bound the solution to the continuous optimization problem by considering the optimal shape of the subgrid for general and more restricted $\sigma_A$ and $\sigma_B$.
In Section~\ref{sec:cont_step}, we derive our main result, (\(\sigma_C(k)\leq \sigma_A(k_A)\sigma_B(k_B)\)) by reducing general subgrids to rectangles.
We also sharpen these lower bounds in Appendix~\ref{appendix:improved_bnd} by instead considering an L-shaped geometry (we reserve these results for the appendix as we do not employ them for any of the applications considered).
}



Equipped with the general bounds derived in Section~\ref{sec:cont_step},
we apply our framework to fast algorithms 
for matrix multiplication, convolution, and partially symmetric tensor 
contractions in Section~\ref{sec:Apps}. 
We obtain lower bounds on both sequential communication (communication 
in a two-level memory hierarchy) as well as parallel communication 
(communication between processors in a distributed-memory computer 
with a fully connected network). The latter bounds can be translated to the
LogGP and BSP model~\cite{solomonik2017communication}. Our lower bounds 
are all novel in that they consider a larger space of algorithms than previous works.
We obtain the first communication lower bounds for nested symmetry 
preserving tensor contraction algorithms~\cite{solomonik2015contracting}, lower bounds 
for multi-dimensional and recursive Toom-Cook (i.e., convolution) that 
match previously known bounds~\cite{bilardi2019complexity,de2019complexity},
and lower bounds for Strassen's 
algorithm, which are asymptotically lower than previous results~\cite{ballard2013graph,bilardi2017complexity}. 
See Table~\ref{tab:CCList} 
for a comparison between previously known lower bounds and the lower bounds derived in this paper.

\begin{table}[h]
    \centering
{\scriptsize
\caption{Communication lower bounds for the Strassen's fast matrix-matrix multiplication and nested Toom-k for 1D convolution. We consider both the sequential model (S) with 
\edits{$M$} size of fast memory and parallel model (P) with \edits{$P$} processors.
A dash indicates we matched the previous lower bound.}
\label{tab:CCList}
\begin{tabular}{lrrrr}
\hline\noalign{\smallskip}
Algorithm & Previous (S) & Previous (P) & This Paper (S) & This Paper (P) \\ \midrule
Strassen's & $\displaystyle \frac{n^{\log_2(7)}}{M^{\log_4(7)-1}}$~\cite{ballard2013graph,bilardi2017complexity} & $\displaystyle \frac{n^{2}}{P^{\log_7(4)}}$~\cite{ballard2013graph} & $\displaystyle \frac{n^{\log_2(7)}}{M^{\log_2(3)-1}}$ (\ref{cor:MMVLB}) & $\displaystyle \frac{n^{\log_3(7)}}{M^{\log_3(2)}}$ (\ref{cor:MMHLB}) \\ \addlinespace[5pt]
\begin{tabular}[x]{@{}l@{}}Recursive\\convolution\end{tabular} & $\displaystyle \frac{n^{\log_k(2k-1)}}{M^{\log_k(2k-1)-1}}$~\cite{bilardi2019complexity} & $\displaystyle \frac{n}{P^{\log_{2k-1}(k)}}$~\cite{de2019complexity} & -- (\ref{cor:ConvVLB}) & -- (\ref{cor:ConvHLB}) \\ 
\noalign{\smallskip}\hline
\end{tabular}
}
\end{table}

\section{Notation, Definitions, and Preliminaries}\label{sec:RkExp}
\subsection{Notational Conventions} \label{sec:notations}
We will denote $\mathbb{N} = \{1,2,\ldots\}$ to be the natural numbers and $\mathbb{R}_+$ as the set of nonnegative reals. For any $n \in \mathbb{N}$, we write $[n] = \{1,2,\ldots, n\}$.

We denote the pseudoinverse of an increasing function $f$ as 
\begin{equation} \label{eq:defpseudoinv}
    \edits{f^{\dagger}(x) = \sup\set{k: f(k) \leq x}}.
\end{equation}
The Kronecker product and Hadamard (entrywise) product are, respectively, $\otimes$ and $\odot$.
\editstwo{
Finally, for convenience we make the following definition.
\begin{defn}\label{def:P}
    Let $\set{\B e_1,\ldots,\B e_n}$ be the standard basis vectors of $\R^n$. For $k \in [n]$, define
    \begin{equation*}
        \mathcal{P}_n^{(k)}
        =
        \set{(\B e_{i_1}|\ldots|\B e_{i_k}) \in \R^{n \times k}: 1\leq i_1 <\ldots <i_k \leq n}.
    \end{equation*}
    In other words, $\mathcal{P}_n^{(k)}$ is the collection of operators $\B P$ such that $\B A\B P$ selects $k$ columns of $\B A$.
\end{defn}
}

\subsection{Bilinear Algorithms} \label{sec:bilinear_alg_def}


\edits{
The target of the communication lower bound framework studied in~\cite{solomonik2015contracting,solomonik2017communication} is analysis of bilinear algorithms. We now formally define these. 
\begin{defn} \label{def:ba_def}
  A \textit{bilinear algorithm} is defined by a matrix triplet,
  \[
  (\B{A},\B{B},\B{C}),
  \]
  (where $\B{A}\in \mathbb{R}^{m_A \times R}$, $\B{B} \in \mathbb{C}^{m_B \times R}$, and $\B{C} \in \mathbb{C}^{m_C \times R}$), which takes in two inputs $\B{x} \in \mathbb{C}^{m_A}$ and $\B{y} \in \mathbb{C}^{m_B}$ and computes an output $\B{z} \in \mathbb{C}^{m_C}$ where
    \begin{align*}
        \B{z} = f(\B{x},\B{y}) = \B{C}\big[ (\B{A}^T\B{x}) \odot (\B{B}^T\B{y}) \big].
    \end{align*}
\end{defn}
}
Here, $R$ is referred to as the \edits{\textit{rank}} of the bilinear algorithm. We refer to the multiplication $\B{A}^T\B{x}$ and $\B{B}^T\B{y}$ as the \textit{encoding step} and the multiplication with $\B{C}$ as the \textit{decoding step}. Similarly, matrices $\B{A}$ and $\B{B}$ may be referred to as the \textit{encoding matrix} while $\B{C}$ is the \textit{decoding matrix}. 

\edits{The power of this framework is the ability to explicitly express the use of recursion in a bilinear algorithm via Kronecker (tensor) products. 
\begin{defn}\label{def:nestedbil}
  A \textit{nested bilinear algorithm} is a bilinear algorithm whose matrix triplet is defined by Kronecker products, i.e.,
    \begin{align*}
        \Big ( \bigotimes_{i=1}^\tau \B{A}_{i}, \bigotimes_{i=1}^\tau \B{B}_{i}, \bigotimes_{i=1}^\tau \B{C}_{i} \Big).
    \end{align*}
\end{defn}
Nested bilinear algorithms are the main tool to develop fast bilinear
algorithms. For example, instead of using the naive eight elementwise
products, Strassen's algorithm for $2 \times 2$ matrix multiplication only
requires seven~\cite{strassen1969gaussian}. Recursively applying the bilinear algorithm via a nested
bilinear algorithm yields the well-known Strassen's algorithm with
$O(n^{\log_2(7)})$ computational complexity. We refer to the 
survey~\cite{pan1984can} for a detailed discussion.} 

Communication lower bounds for bilinear algorithms can be reasoned by exploiting the sparsity pattern of the encoding and decoding matrices. This approach has found success for the naive $O(n^3)$ matrix multiplication by utilizing volumetric arguments, such as the H\"{o}lder-Brascamp-Lieb inequalities~\cite{holder1889uber,brascamp1976best}, since the sparsity structure of encoding and decoding matrices from the naive algorithm is relatively simple, i.e., one nonzero per column of the matrix. In contrast, the respective matrices for Strassen's algorithm have a dense structure, so volumetric inequalities do not solely suffice here.

\subsection{Communication Cost Lower Bounds for Bilinear Algorithms} \label{sec:commlb_frmwrk}
In this paper, we derive lower bounds on communication in two standard settings: sequential and parallel.
\edits{In the sequential setting, we consider a fast but small memory of size $M$ (e.g., cache),
where computation is performed, and slow but large memory (e.g., main memory). 
In the parallel setting, we consider $P$ processors that communicate over a network and bound the largest amount of data sent or received by any of the $P$ processors.}
In both cases, we quantify communication cost in the number of data elements, which may be elements of the input or output, as well as intermediate values (partial sums of inputs or of products of inputs). Finally, our computational model prohibits recomputation in the bilinear algorithm, meaning that if a bilinear product has already been computed and is not in fast memory or the current processor, the algorithm must communicate the value instead of recomputing it~\cite{solomonik2017communication}. \editstwo{This restriction stems from the assumption in proofs of the general lower bounds framework we build on (specifically, Lemma 5.2 in~\cite{solomonik2017communication})}.
Extension of these communication lower bounds to permit recomputation is of interest and has been explored for related problems~\cite{jia1981complexity,bilardi2019complexity,de2019complexity,bilardi2017complexity}. 
For integer multiplication and matrix multiplication, past works have shown that the best known lower bounds hold even when allowing recomputation~\cite{bilardi2019complexity,de2019complexity,bilardi2017complexity,nissim2019revisiting}.

\edits{Lower bounds for both settings are studied in the work of~\cite{solomonik2017communication}. They are derived by establishing the expansion bound, which we define next. Recall that $R$ is the associated rank of a bilinear algorithm and \editstwo{$\mathcal{P}^{(k)}_R$ is the collection of operators selecting $k$ columns from a matrix with $R$ columns (Definition~\ref{def:P})}. 
\begin{restatable}{defn}{expbdrest} \label{def:ExpansionBnd}
    \sloppy \editstwo{The bilinear algorithm $(\B{A},\B{B},\B{C})$ with rank $R$ has \textit{expansion} $\mathcal{E}^* : \mathbb{N}^3 \to \mathbb{N}$ if
    for any triplet $k_A, k_B, k_C \in [R]$,
    \begin{align*}
      &\mathcal{E}^*(k_A, k_B, k_C)  \\
      &= 
      \max_{\B P \in \mathcal{P}^{(k)}_R, \forall k} \{ 
      \#\text{cols}(\B{P}) : \mathrm{rank}(\B{AP}) \leq k_A, \mathrm{rank}(\B{BP}) \leq k_B, \mathrm{rank}(\B{CP}) \leq k_C \}.
    \end{align*}}
    Likewise, the same bilinear algorithm has non-decreasing (in all variables) \textit{expansion bound} $\mathcal{E} : \mathbb{N}^3 \mapsto \mathbb{R}$ if $\mathcal{E}$ upper bounds $\mathcal{E}^*$, i.e., 
    \[
        \mathcal{E}^*(k_A,k_B,k_C) \leq \mathcal{E}(k_A,k_B,k_C), \ \forall k_A,k_B,k_C \in [R].
    \]
\end{restatable}}
We relaxed the expansion bound from~\cite{solomonik2015contracting,solomonik2017communication} to be from $\mathbb{N}^3 \mapsto \mathbb{R}$ instead of $\mathbb{N}^3 \mapsto \mathbb{N}$, which can be done without loss of generality by rounding down non-integer values. The expansion \editstwo{function quantifies the largest subset of the bilinear algorithm we can complete when given a limited set of data, whereas the expansion bound is an easier-to-compute upper bound on the expansion function}. These functions are closely related to the \textit{edge expansion} of a graph $G$, a function previously used to derive communication lower bounds via the graph expansion framework~\cite{ballard2013graph,jia1981complexity,bilardi2019complexity,bilardi1999processor}. The edge expansion helps to lower bound the number of edges leaving any sufficiently small subgraph of $G$ relative to the number of vertices in the subgraph. Similarly, the expansion bound helps to lower bound the rank of any subset of a bilinear algorithm relative to the size of the subset.

\edits{Equipped with this quantity,~\cite{solomonik2017communication} develops communication lower bounds in both the sequential and parallel setting. We start with the sequential setting.
\begin{restatable}[Theorem 5.3~\cite{solomonik2017communication}]{proposition}{vlbrest} \label{prop:edgar_seq_thm}
  Given a bilinear algorithm $(\B{A},\B{B},\B{C})$ (where $\B{A} \in \mathbb{C}^{m_A \times R}$, $\B{B} \in \mathbb{C}^{m_B \times R}$, and $\B{C} \in \mathbb{R}^{m_C \times R}$), a corresponding expansion bound function $\mathcal{E}$, and a sequential model with fast memory of size $M$, then any procedure for computing the bilinear algorithm must communicate at least
  \begin{align*}
     \max\Big\{ \frac{2RM}{\mathcal{E}^{\max}(M)}, m_A + m_B + m_C \Big\},
  \end{align*}
  data elements between fast and slow memory, where
  \begin{align*}
    \mathcal{E}^{\max}(M) = \max_{\substack{r^{(A)},r^{(B)},r^{(C)} \in \mathbb{N}\\ r^{(A)} + r^{(B)} + r^{(C)} =3M}} \Big\{ \mathcal{E}\big(r^{(A)},r^{(B)},r^{(C)}\big) \Big\}
  \end{align*}
  and $\mathcal{E}^{\max}(M)$ is assumed to be increasing and convex. 
\end{restatable}
The second term $m_A + m_B + m_C$ accounts for the reading of the input and writing of the output. Next, we state the parallel communication lower bound. To be consistent with~\cite{solomonik2017communication}, we assume the algorithm is \textit{storage-balanced}: at the beginning of the algorithm, the input is evenly distributed among $P \in \mathbb{N}$ processors. In the end, the output is evenly distributed among the processors.
\begin{restatable}[Theorem 5.4~\cite{solomonik2017communication}]{proposition}{hlbrest} \label{prop:edgar_par_thm}
  Given a bilinear algorithm $(\B{A},\B{B},\B{C})$ (where $\B{A} \in \mathbb{C}^{m_A \times R}$, $\B{B} \in \mathbb{C}^{m_B \times R}$, and $\B{C} \in \mathbb{R}^{m_C \times R}$), a corresponding expansion bound function $\mathcal{E}$, and a parallel model with $P$ processors, then any storage-balanced procedure for computing the bilinear algorithm must communicate at least
  \begin{align*}
      r^{(A)} + r^{(B)} + r^{(C)}
  \end{align*}
  data elements between processors, where $r^{(A)},r^{(B)},r^{(C)} \in \mathbb{N}$ satisfy
  \begin{align*}
      \frac{R}{P} 
      \leq
      \mathcal{E}\bigg( r^{(A)} + \frac{m_A}{P}, r^{(B)} + \frac{m_B}{P}, r^{(C)} + \frac{m_C}{P} \bigg).
  \end{align*}
\end{restatable}
Thus, expansion lower bound $\mathcal{E}$ yields lower bounds on sequential and parallel communication cost directly via these two propositions.
}

\subsection{Rank Expansion Bounds}

To characterize the expansion lower bound needed for our communication bounds ($\mathcal{E}$), we consider the rank of subsets of columns of each individual encoding/decoding matrix. 
\editstwo{Recall that $\B{P} \in \mathcal{P}^{(k)}_n$ selects a subset of $k$ columns.}
\begin{defn} \label{def:RankExpDef}
    The \textit{rank expansion} of $\B{A} \in \mathbb{C}^{m \times n}$,
    \edits{$\tilde{\gs}  :  [n] \mapsto \mathbb{N}$}, is defined as
    $\tilde{\gs}(k) = \displaystyle \min_{\B{P} \in \mathcal{P}^{(k)}_n} \big\{ \rank(\B{AP})\big\}$.
\end{defn}
\edits{
Let $\tilde{\gs}_A$, $\tilde{\gs}_B$, and $\tilde{\gs}_C$ be rank expansion functions for matrices $\B{A}$, $\B{B}$, and $\B{C}$, respectively. 
Since, for any $\B P\in\mathcal{P}^{(k)}_n$, $\tilde{\gs}_A(\#\mathrm{col}(\B{P}))\leq \mathrm{rank}(\B{AP})$ (and similar for $\B B$, $\B C$), these rank expansion functions yield an expansion bound $\mathcal{E}$ (Definition~\ref{def:ExpansionBnd}) of the form, 
\begin{equation}  \label{eq:rnk_exp_to_exp_bnd}
  \mathcal{E}( r^{(A)}, r^{(B)}, r^{(C)}) := \min\{ \tilde{\gs}_A^\dagger( r^{(A)}), \tilde{\gs}_B^\dagger(r^{(B)}), \tilde{\gs}_C^\dagger(r^{(C)})\},
\end{equation}
for any $r^{(A)},r^{(B)},r^{(C)} \in \mathbb{N}$,
where recall the dagger is the pseudoinverse defined in~\eqref{eq:defpseudoinv}.
}
\editstwo{
While relatively easy to compute, the proposed expansion bound above may not be tight with the expansion function $\mathcal{E}$, as shown in the example below.}

\begin{example} \label{ex:eqgp_1}
    \editstwo{Consider the bilinear algorithm $(\B{A}, \B{B}, \B{C})$ for standard matrix multiplication of two $n \times n$ matrices. Matrices $\B{A}$, $\B{B}$, and $\B{C}$ each contain $n$ copies of each column from an identity matrix of dimension $n^2$ (see~\cite[Lemma B.1]{solomonik2017communication} for an idea on the sparsity pattern). Then for any $r^{(A)},r^{(B)},r^{(C)} \in [n^2]$, the Loomis-Whitney inequality~\cite[Lemma B.1]{solomonik2017communication} produces an expansion bound
    \[
        \mathcal{E}_{\text{LW}} (r^{(A)}, r^{(B)}, r^{(C)}) :=
        \sqrt{r^{(A)}r^{(B)}r^{(C)}},
    \]
     while one can show~\eqref{eq:rnk_exp_to_exp_bnd} simplifies to 
     \[
        \mathcal{E}_{\text{RankExp}} (r^{(A)},r^{(B)},r^{(C)}) 
        := 
        n \cdot \min\{r^{(A)}, r^{(B)}, r^{(C)}\}.
    \]
    When the fast memory size is not large, i.e., $\max\{r^{(A)}, r^{(B)}, r^{(C)}\} \ll n^2$, then $\mathcal{E}_{\text{LW}}$ is tighter. The gap between the two expansion bounds arises because a subset of columns from one matrix, say $\B A$, may be low rank while the same subset for another matrix, say $\B B$, may be nearly full rank. The Loomis-Whitney inequality seems to account for this low-rank property while~\eqref{eq:rnk_exp_to_exp_bnd} does not.}
\end{example}
However, the main advantage of the proposed expansion bound~\eqref{eq:rnk_exp_to_exp_bnd} is that it can derive bounds for any bilinear algorithm. In contrast, the Loomis-Whitney inequality requires simple access patterns~\cite{ballard2011minimizing}, which excludes its application for fast algorithms such as Strassen's algorithm~\cite{strassen1969gaussian}.

As defined, the rank expansion is a discrete function, which makes it challenging to derive simple closed-form expressions for its pseudoinverse. Instead, we seek to lower bound $\tilde{\gs}$ by a continuous increasing function $\gs$.

\begin{defn}\label{def:RankExpLBDef}
    Let $\B{A} \in \mathbb{C}^{m \times n}$ with rank expansion $\tilde{\gs}$.    
    A \textit{rank expansion lower bound} $\sigma$ for $\B{A}$ is a continuous, nonnegative, and increasing function $\sigma$ on $\R_+$ such that
    $\sigma(k) \leq \tilde{\sigma}(k)$ for all $k \in [n]$.
\end{defn}

\edits{
For a nested bilinear algorithm (Definition~\ref{def:nestedbil}), we seek to obtain ${\gs}_A$ from the rank expansion ${\gs}_{i}$ of each term $\B A_i$ in the Kronecker product defining $\B A$ (and similar for $\B B$, $\B C$).
Often ${\gs}_{i}$ is easy to obtain, e.g., for Strassen's algorithm each $\B A_i$ is the same $4$-by-$7$ matrix, while in other cases each $\B A_i$ may vary in size but \editstwo{may have} a known/simple rank expansion.
}

\section{Summary of General Lower Bound Results} \label{sec:main_results}

As motivated above, we seek a rank expansion lower bound $\gs_C$ for $\B{C} = \B{A} \otimes \B{B}$ given rank expansion lower bounds $\gs_A$ and $\gs_B$.
We state our first main result below. Recall the pseudoinverse from~\eqref{eq:defpseudoinv}, denoted with a dagger.

\begin{restatable}{theorem}{mainthmone} \label{thm:lb}
    Suppose $\sigma_A$ and $\sigma_B$ are concave rank expansions lower bounds for $\bA \in \mathbb{C}^{m_A \times n_A}$ and $\bB\in \mathbb{C}^{m_B \times n_B}$, respectively, and $\gs_A(0) = \gs_B(0) = 0$.
    Let $d_A = \gs_A^{\dagger}(1)$, $d_B = \gs_B^{\dagger}(1)$. Then
    \begin{equation*}
        \gs_C(k) = \min_{\substack{k_A\geq d_A,~k_B \geq d_B\\ 
        k_A k_B \geq k}}\gs_A(k_A)\cdot\gs_B(k_B)
    \end{equation*}
    is a rank expansion lower bound for $\B C = \B A \otimes \B B$.  
\end{restatable}

The proof of this result is lengthy. We delay it \editstwo{until} Section~\ref{sec:main_proof_1}. 
Recursively applying the above theorem to $\B{C} = \bigotimes_{i=1}^p \B{A}_i$, where $p \geq 3$, requires solving a non-trivial optimization problem. To derive simple rank expansion lower bounds for nested bilinear algorithms, we require the rank expansion lower bounds $\gs_i$ for each $\B{A}_i$ to be log-log concave.

\begin{defn} \label{def:loglog}
    We say function $f(x)$ is \textit{log-log concave} (resp. \textit{convex}) if $\ln(f)$ is concave (resp. convex) in $\ln(x)$.
\end{defn}

The class of log-log concave functions contains many functions that can serve as tight rank expansion lower bounds, 
such as logarithms and polynomials with a leading term that has an exponent greater than or equal to 1. 
For further discussions, see Section~\ref{sec:main_proof_2}. 
For log-log concave rank expansion functions, we obtain the following bound for nested bilinear algorithms.


\begin{restatable}{theorem}{mainthmtwo} \label{thm:nest_lb}
    Let $\gs_i$ be a rank expansion lower bound of $\bA_i \in \R^{m_i \times n_i}$ for $i = 1,2,\ldots,p$ ($p\geq 2$), and let $\bC = \bigotimes_{i = 1}^p \bA_i$.
    If $\gs_i$ are concave and log-log concave and satisfy $\gs_i(0) = 0$, then
    \begin{equation} \label{eq:lb_for_many}
        \gs_C(k) = \min_j\set{\gs_j\left(\frac{k}{\prod_{i \neq j}d_i}\right)},
    \end{equation}
    where $d_i = \gs_i^{\dagger}(1)$, is a concave and log-log concave rank expansion lower bound of $\bC$. In particular, if $\gs_i(k) = (k / k_i)^{q_i}$ with $k_i\geq 1$ and $q_i \in (0, 1]$, then
    \begin{equation}\label{eq:lb_for_poly}
        \gs_C(k) = \left(\frac{k}{\prod_i k_i}\right)^{\min_j q_j}.
    \end{equation}
    If $\gs_i(k) = a_i \ln (b_i k + 1)$, then
    \begin{equation}\label{eq:lb_for_log}
        \gs_C(k) = a \ln (bk + 1),
    \end{equation}
    where $a = \min_i a_i$ and $\displaystyle b = \frac{e^{1/a}-1}{\prod_{i=1}^{n} b_i^{-1}(e^{1/a_i}-1)}$.
\end{restatable}
The proof is given in Section~\ref{sec:main_proof_2}. 
In Section~\ref{sec:Apps}, we illustrate the application of these results and derive new lower bounds for the communication cost of three bilinear algorithms (see Section~\ref{sec:intro} for a high-level description of our application contributions).
Theorems \ref{thm:lb} and \ref{thm:nest_lb} are easy to use and capable of producing nontrivial tight lower bounds.
However, one potential problem is that they both require defining the lower bound $\gs$ on the entire range $\R_+$, which requires extrapolating $\gs$ beyond the intended domain $[0, n]$ (e.g., in Theorem~\ref{thm:lb}, the optimal $k_A,k_B$ may have $k_A\geq n_A$ or $k_B \geq n_B$).
In 
Appendix~\ref{appendix:improved_bnd},
we derive lower bounds $\gs_C$ that do not require extrapolating $\gs_A$ and $\gs_B$ to $\R_+$, which often results in tighter lower bounds on $\gs_C$.

\section{Rank Analysis of Kronecker Product via Grid Expansion} \label{sec:NestedSigma}

\edits{
In this section, we introduce a grid representation to analyze the rank of a column-wise submatrix $\bC\bP$ in $\bC = \bA \otimes \bB$. 
The main idea is to represent columns of $\bC\bP$ as a subset in the 2-D grid representing columns of $\bA$ and $\bB$.
We then manipulate the grid representation, compactifying the set of grid points, while keeping intact any low-rank structure.
After the combinatorial arguments in this section, we derive rank bounds (including our main result, Theorem~\ref{thm:lb}) by continuous analysis of the resulting compact geometric structure formed by the grid points.
}

\subsection{Grid Framework}
\label{sec:prelim}


\edits{
We first introduce our grid representation and then define the notions of a grid basis and a compact grid.
A basis is a set of grid points that represents a linearly independent set of columns that span all points in a given grid.
For compact grids, we show that there is a basis with a simple reducible structure.
}
\subsubsection{Grid Representation}

Let $\B{A} \in \mathbb{C}^{m_A \times n_A}$ and $\B{B} \in \mathbb{C}^{m_B \times n_B}$ be arbitrary matrices, and let $\B{C} = \B{A} \otimes \B{B}$ be given by
\[
    \B{C} 
    = 
    \begin{bmatrix}
        \ba_{1} \otimes \B{B} \ & \ \cdots & \ \ba_{i} \otimes \B{B} \ \cdots \ \ba_{n_A} \otimes \B{B} 
    \end{bmatrix}.
\]
The column $\B{c}_k$ from $\B{C}$ is defined as $\ba_i \otimes \bb_j$ for some columns $\ba_i$ and $\bb_j$. Thus, we will refer to a column $\B{c}_k$ by the tuple $(i, j)$. 

Now recall that $\mathcal{P}_n^{(k)}$ is the set of matrices comprised of $k$ different columns from an identity matrix of size $n$.
For any $\B{P} \in \mathcal{P}_{\#\text{cols}(\B{C})}^{(k)}$, $\bC\bP$ then 
contains $k$ column vectors of $\bC$. 
Thus, $\bC\bP$ can be identified (up to a column reordering) by a
set of $k$ tuples. 
In particular, we can view this set of tuples as a set of grid points from an $n_A \times n_B$ grid, as shown in Figure \ref{fig:example}. 
Note that we use the \textit{Cartesian indexing system}, not the array indexing system!
We call this set of grid points \textit{the grid representation of} $\bC\bP$. 
\begin{figure}[h]
   \includegraphics[width=8cm]{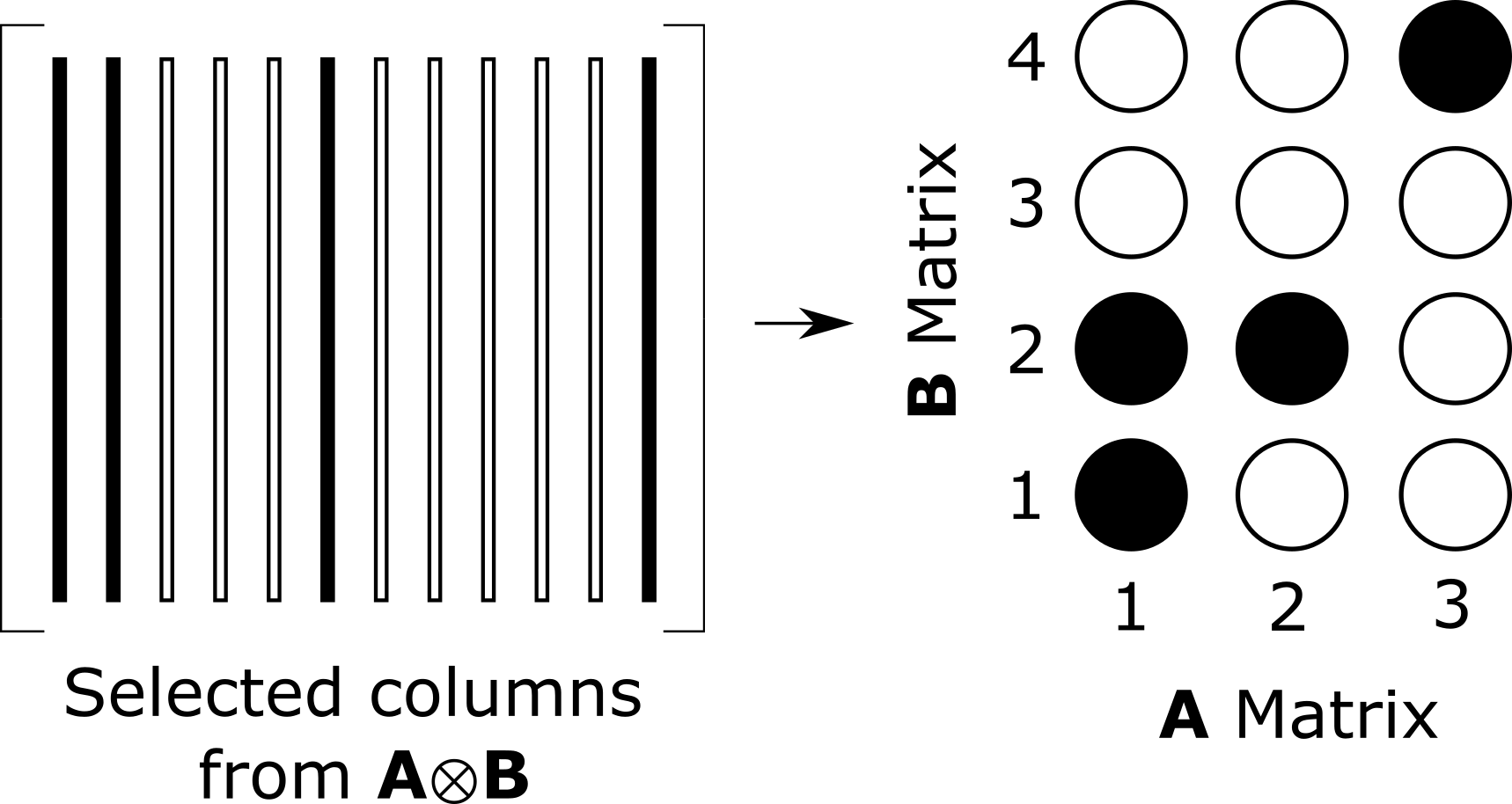}
   \centering
   \caption{
    Let $\B{A} \in \mathbb{C}^{m_A \times 3}$, $\B{B} \in \mathbb{C}^{m_B \times 4}$, and $\bC = \bA \otimes \bB$. The subset of columns $\{\ba_1 \otimes \bb_1, \ba_1 \otimes \bb_2, \ba_2 \otimes \bb_2, \ba_3 \otimes \bb_4\}$ from $\bC$ is represented by the black-filled columns on the figure on the left. This subset of columns is equivalent to the matrix product $\B{CP}$ for some $\B{P} \in \mathcal{P}_{12}^{(4)}$. The grid representation of $\B{CP}$ is the $3 \times 4$ grid on the right. Each black-filled circle represents a selected column of $\B{C}$, and its $(i,j)$ index is the column from $\B{A}$ and $\B{B}$ defining that column.}
    \label{fig:example}
\end{figure}

For a grid $G \sset[n_A]\times[n_B]$, 
we write $G_{[i,\cdot]} = G \cap (\set{i} \times [n_B])$ and
$G_{[\cdot,j]} = G \cap ([n_A] \times \set{j})$ 
to be the set of grids points in column $i$ and row $j$ of $G$, respectively.
We denote the size of a grid $G$ by $|G|$.
When it is not ambiguous,
we use $(i, j)$ and $\ba_i \otimes \bb_j$ interchangeably to denote a column of $\bC$, and associate a grid $G$ with a submatrix of $\bC$, referring to the subset of column vectors represented by $G$.
When we say $\spp\set{G}$, we mean the space spanned by column vectors represented by grid $G$, and we call the dimension of this space by $\rank(G)$.

We endow the grid points with the colexicographic order, 
where $(i,j)$ precedes $(i',j')$ if $i < i'$ or if $i=i'$ and $j < j'$. 
We then denote by $[(i, j)]$ the set of points $(k, \ell) \leq (i, j)$ in this order, and we write $((i, j)) = [(i, j)] \setminus \set{(i, j)}$.
We naturally extend this ordering to two grids by sorting their points and then comparing the first pair, second pair, and so on. 

In order to analyze the rank of $\bC\bP$, we introduce the notion of a basis, the maximal linearly independent set of columns in $\bC\bP$, as defined below.

\begin{defn}
    Let $G$ be the grid representation of $\B{CP}$.
    The \textit{basis} of $G$ is a subgrid $B\sset G$,
    such that column vectors in $B$ comprise a maximal linearly independent set of column vectors in $G$,
    and $B$ is minimal in the colexicographic order among all such subgrids of $G$.
\end{defn}

Algorithm~\ref{alg:basis_construction} illustrates this definition and is one concrete way to find the basis.
We traverse the grid in the colexicographic order and add point $(p, q)$ to the basis if it is not in the span of $((p,q))$.

\begin{algorithm}[t]
   \begin{algorithmic}[1]
   \Function{BasisSelection}{Grid $G$}
     \State $B \gets \{ \ \}$ 
     \For{$(i,j) \in G$ using colexicographic traversal order} 
       \If{$(i, j) \not \in \mathrm{span}\set{B}$} \label{line:in}
         \State $B \gets B \cup \{(i,j)\}$ \Comment{Add grid point if it is ``new'' to the span}
       \EndIf
     \EndFor
   \EndFunction
   \Return $B$
   \end{algorithmic}
  \caption{Constructing a unique basis}
   \label{alg:basis_construction}
\end{algorithm}

In order to describe which columns in $\bA$ and $\bB$ are involved in $\bC\bP$ and simplify the notations, we define the following projections.

\begin{defn}
  The $\B{A}$-\textit{projection} of the grid $G$, denoted by $P_A(G$),
  is the set of indices 
  $I \subseteq [n_A]$ such that $\csg{G}{i} \neq \emptyset$. Likewise, the 
  $\B{B}$-\textit{projection}, $P_B(G)$,
  is the set of indices $J \subseteq [n_B]$
  such that $\rsg{G}{j} \neq \emptyset$.
\end{defn}
The $\B{A}$-projection and $\B{B}$-projection can be viewed as the shadow of the grid $G$ onto the $x$-axis (for columns of $\B{A}$) and the $y$-axis (for columns of $\B{B}$), respectively. We provide an example in Figure~\ref{fig:shadow}.

\begin{figure}[htbp]
   \includegraphics[width=0.6\linewidth]{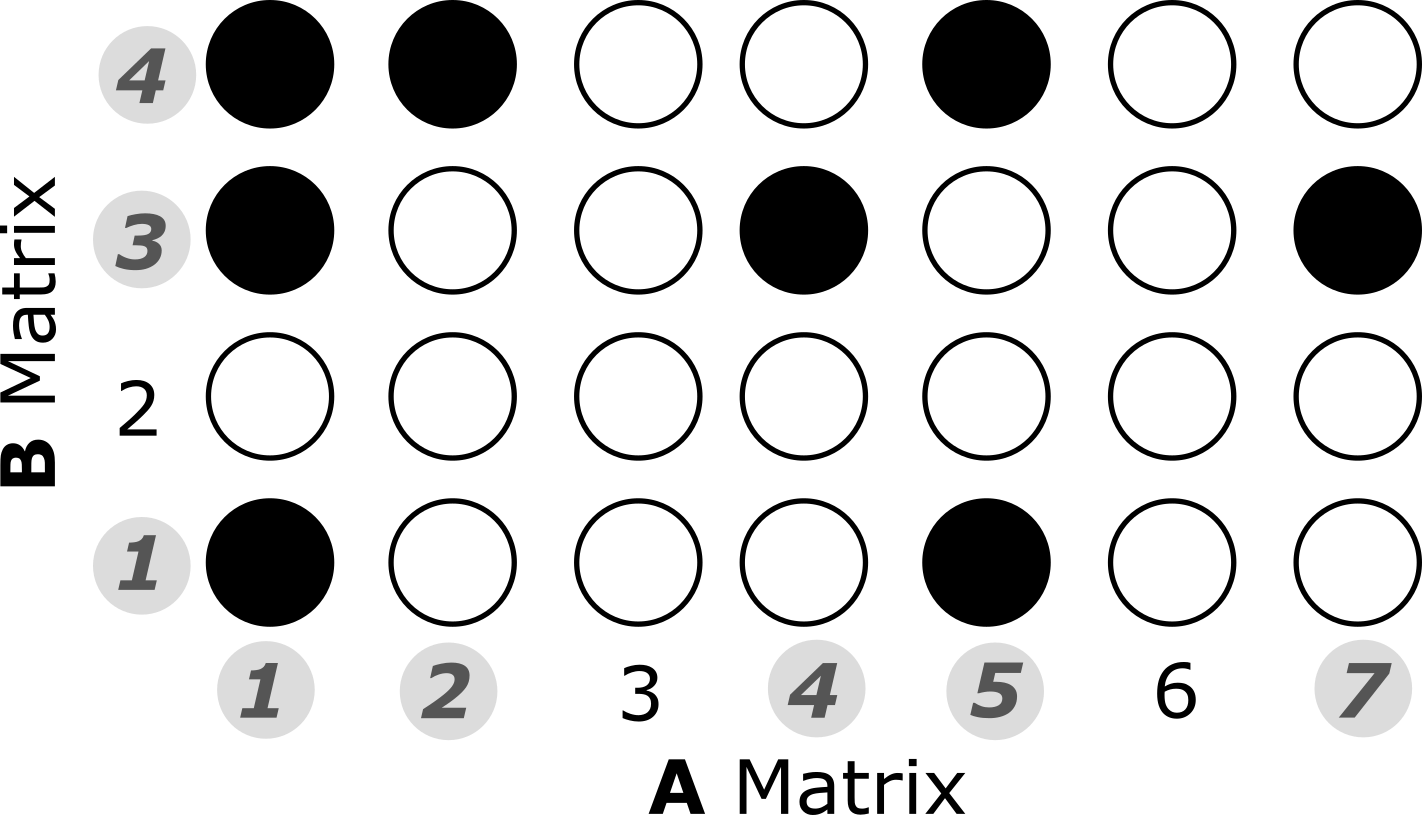}
   \centering
   \caption{A grid representation for a $\B{CP}$ where $\B{C} = \B{A} \otimes \B{B}$. The $\B{A}$-projection is the set $\{1,2,4,5,7\}$, or the set of indices on the x-axis highlighted in bold italics and with a gray ball. Similarly, the $\B{B}$-projection is the set $\{1,3,4\}$.}
    \label{fig:shadow}
\end{figure}

\subsubsection{Compact Dense Grids}

Now we introduce a class of structured grids that is the key to our analysis:
compact dense grids (CDGs).

\begin{defn} \label{def:cdg}
    A grid $G$ is a \textit{dense grid} if for every 
    $i$, $\csg{G}{i} = \set{i}\times [k]$ for some $k\geq 0$.
    A grid $G$ is a \textit{compact dense grid} (CDG) if it is a dense grid and 
    $|\csg{G}{i}|$ is non-increasing in $i$.
\end{defn}

We show the difference between a non-dense grid, a dense grid, and a CDG in
Figure~\ref{fig:compact}. We can transform any arbitrary grid into a dense grid 
by collapsing, i.e., letting the grid points fall vertically, as shown in the same figure. We make this concrete below.

\begin{defn} \label{def:vcol}
    Let $G$ be an arbitrary grid. Then a \textit{vertical collapse} (VCollapse) of
    $G$ produces a dense grid ${D}$, where
    \begin{align*}
        {D}_{[i,\cdot]} = 
        \begin{cases}
        \set{i} \times \big{[}|G_{[i,\cdot]}|\big{]} \ &: G_{[i,\cdot]} \neq \emptyset \\
        \emptyset \ &: G_{[i,\cdot]} = \emptyset
        \end{cases}
    \end{align*}
\end{defn}

Given any grid $G$, we can vertically collapse it to a dense grid $D_0$.
Then with a reordering of the columns of matrix $\bA$ (the $x$-axis), we can produce a CDG $D$ from $D_0$.
This process is illustrated in Figure~\ref{fig:compact}.
These two operations commute, and note that only the VCollapse step may change the rank of $G$ (or equivalently $\bC\bP$).

\begin{figure}[h]
   \includegraphics[width=\linewidth]{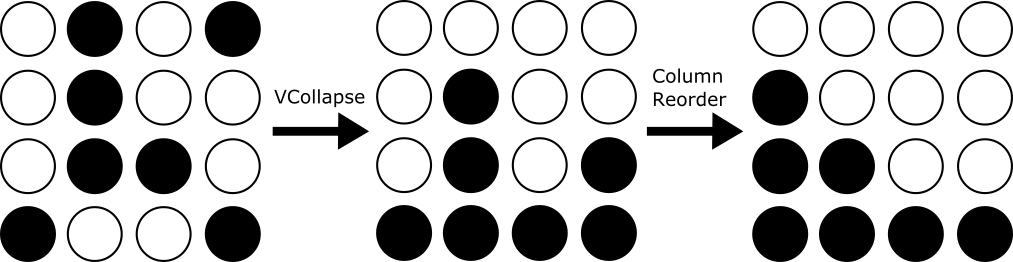}
   \centering
   \caption{Example a non-dense grid (left), dense grid (center), and CDG
   (right). The left grid is non-dense since the $\B{B}$-projection of the second column equal is $\{2,3,4\}$. 
   The action of a vertical collapse ensures the resulting grid (center) is dense.}
    \label{fig:compact}
\end{figure}

\subsubsection{Reducible Structure in the Basis of a Compact Dense Grid}

Here we establish the most important property of a CDG -- its basis has a reducible structure. 
This is one of the key steps towards the main theorems.
We start with two lemmas.
As before, we denote the columns of $\bA$ and $\bB$ by $\set{\ba_i}$ and $\set{\bb_j}$.
We identify column vectors of $\bC = \bA \otimes \bB$ in the 2-dimensional grid.

\begin{lemma} \label{lem:inspan}
    Let $\bu \in \mathbb{C}^n$, $\bv \in \mathbb{C}^m$. 
    Let $\set{\ba_i}_{i = 1}^p \sset \C^n$, 
    $\set{\bb_j}_{j = 1}^q \sset \C^m$ be sets of vectors.
    Then $\bu \otimes \bv \in \spp\set{\ba_i \otimes \bb_j}$
    if and only if $\bu \in \spp\set{\ba_i}$ and 
    $\bv \in \spp\set{\bb_j}$.
\end{lemma}

\begin{proof}
    \sloppy \edits{The result follows from the basis for a tensor product of finite-dimensional vector spaces. See for instance~\cite{halmos2017finite}.}
\end{proof}

\begin{lemma} \label{lem:newcol}
    If $\ba_p \notin \spp\set{\ba_i}_{i<p}$ and $F\sset[(p, q)]$ is a set such that $(p, q) \in \spp\set{F}$, then $(p, q) \in \spp\set{F_{[p, \cdot]}}$.
\end{lemma}

\begin{proof}
    Let $\csg{F}{p} = \set{(p, j_m)}_m$. 
    Let $V' = \spp\set{\bb_{j_m}}$ and $V = \spp\set{\ba_p \otimes \bb_{j_m}}$ $= \ba_p \otimes V'$. Let $\bP_V$, $\bP_{V'}$ be orthogonal projections. 
    Let $\bu = (\bI - \bP_V)(\ba_p \otimes \bb_q) = \ba_p \otimes (\bI - \bP_{V'}) \bb_{q}$. 
    If $\bu \neq 0$, then $\bu$ is in the span of $F\setminus\set{(p, j_m)}_m$, but this is impossible by Lemma \ref{lem:inspan}, since $\ba_p \notin \spp\set{\ba_i}_{i<p}$. 
    Thus, $\bu = 0$ and $(p, q) \in \spp\set{F_{[p, \cdot]}}$.
\end{proof}

Now we are ready to prove the following key result on the structure of the basis of a CDG.

\begin{proposition} \label{prop:tnsprodB}
    Let $D$ be a CDG and its basis be $B_D$. Let the $\bA$ and $\bB$-projection of $B_D$ be $X = P_A(B_D)$ and $Y = P_B(B_D)$, respectively. Then
    \begin{enumerate}[nosep, label=(\roman*)]
        \item $(1, j) \in B_D$ if and only if $\bb_j \notin \spp\set{\bb_k}_{k \in Y \cap [j - 1]}$,
        \item $(i, 1) \in B_D$ if and only if $\ba_i \notin \spp\set{\ba_k}_{k \in X \cap [i - 1]}$,
        \item $B_D = (X \times Y) \cap D$,
        \item For each $p \in X$, $\csg{(B_D)}{p}$ spans $\csg{D}{p}$,
        \item For each $q \in Y$, $\rsg{(B_D)}{q}$ spans $\rsg{D}{q}$.
    \end{enumerate}
\end{proposition}

\begin{proof}
    Part (i) is clear since the basis has minimal colexicographic order (see Algorithm \ref{alg:basis_construction}). 
    
    For part (ii), the ``if" direction follows from Lemma \ref{lem:inspan}, which implies $(i, 1)$ is not spanned by preceding columns $((i, 1))$, 
    and we conclude using the minimal colexicographic order property of the basis.
    For the ``only if" part, note that $(i, 1) \in B_D$ implies $(i, 1) \notin \spp\set{(k, 1)}_{k \in X \cap[i-1]}$.
    The conclusion then follows.
        
    For part (iii) if $D$ only involves a single column of the 2-D grid, then it is clear.
    Now suppose that CDG $D$ involves at least $p > 1$ columns of the 2-D grid.

    \textit{Case 1}: If $\ba_p$ is spanned by $\set{\ba_i}_{i < p}$, then $\fall q$, $(p, q)$ is spanned by preceding points $\set{(i, q)}_{i < p} \sset ((p, q)) \sset D$, where the last inclusion is from the CDG structure of $D$. Thus, $(p, q)$ is not added to $B_D$ for any $q$.

    \textit{Case 2}: Now suppose $\ba_p$ is \textit{not} spanned by $\set{\ba_i}_{i < p}$. 
    For a fixed $q$, if $\bb_q \notin \spp\set{\bb_j}_{j < q}$, then by Lemma \ref{lem:newcol}, $(p, q)$ is not spanned by preceding columns $((p, q))$ and is thus added to the basis. 
    If $\bb_q \in \spp\set{\bb_j}_{j < q}$, then $(p, q)$ is spanned by $\set{(p, j)}_{j < q} \sset ((p, q)) \sset D$. Thus $(p, q) \notin B_D$.

    Case 1 means that if $\ba_p$ is spanned by $\set{\ba_i}_{i < p}$, then we skip this column in forming $B_D$. Case 2 implies that if we do not skip the $p$th column, then we add point $(p, q)$ to $B_D$ if and only if $\bb_q \notin \spp\set{\bb_j}_{j < q}$ if and only if $(1, q) \in B_D$. Thus, $B_D = (X \times Y) \cap D$.

    Parts (iv) and (v) are then clear.
    Part (iv) follows from (i), (iii) and
    part (v) follows from (ii), (iii).  
    This completes the proof.
\end{proof}

Graphically, this proposition says that $B_D$ should look like the one in Figure \ref{fig.DBD}, and for non-empty rows (columns) in $B_D$, the selected vectors span the vectors in that row (column) of $D$. 

\begin{figure}
    \centering
    \includegraphics[width = 0.45\linewidth]{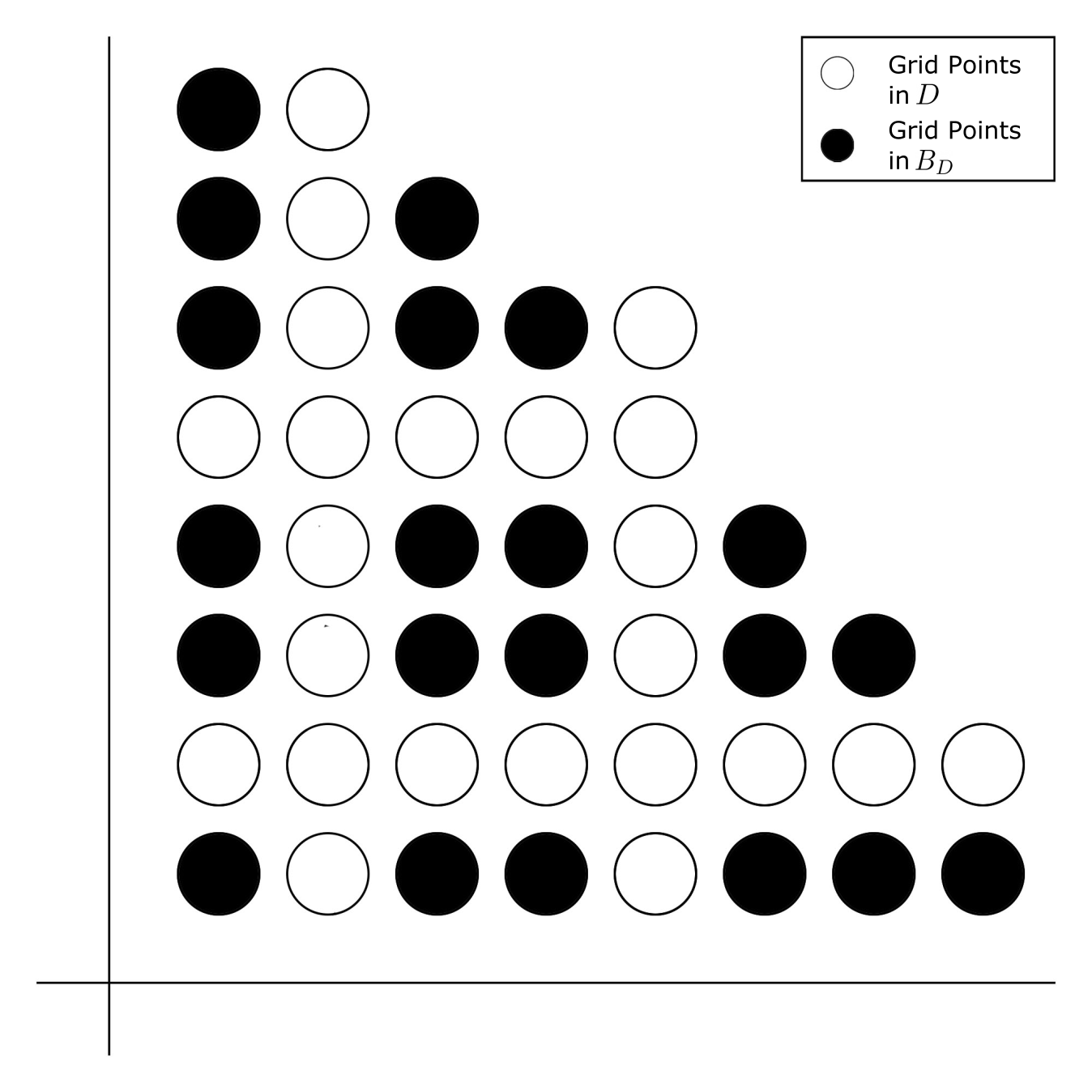}
    \caption{An illustration of the structure of $B_D$.}
    \label{fig.DBD}
\end{figure}

\subsection{Grid Compactification and Expansion} \label{sec:main_proof}

\edits{
Having shown in Proposition~\ref{prop:tnsprodB} that the basis for a CDG has a simple structure, our goal is now to bound the basis size (rank) of an arbitrary grid using that of a CDG constructed by collapsing that grid (via VCollapse, see Definition~\ref{def:vcol}).
To achieve this, we first define a notion of CDG expansion, 
which serves to upper bound the number of columns in $\bC$ spanned by a basis of a CDG shape.
We then show that a subset of the basis of a collapsed CDG may be used to lower bound the rank of an arbitrary grid while having an expansion that upper bounds the size of the grid.
That is, for any grid $G$, we find a CDG $S$ such that
\begin{equation*}
    |S| \leq \rank(G) \text{~~and~~} |G| \leq |\mathrm{GridExp}(S)|. 
\end{equation*} 
It then suffices to bound $|\mathrm{GridExp}(S)|$ with $|S|$.
Since the expansion we consider preserves a simple geometric structure, we achieve this \editstwo{latter bound} by continuous analysis in Section~\ref{sec:cont_step}.
}

\subsubsection{Basis Expansion} \label{sec:discrete_analysis}

We start by defining $\spcexp$, or the grid expansion, on CDGs.
For a CDG-shaped basis, this upper bounds the number of columns in $\bC$ the basis can span.
The definition consists of two steps:
$\vspcexp$ and $\hspcexp$. They commute with each other. 

\begin{defn} \label{def:gridexp_def}
    Let $\gs_A$ and $\gs_B$ be rank expansion lower bounds of $\bA$ and $\bB$.
    For a CDG $S$, $\vspcexp(S)$ is a CDG with the same number of columns of $S$. Each column is of size
    \begin{equation*}
        |\csg{\vspcexp(S)}{i}| = \min\set{n_B,~\lfloor\gs_B^\dagger(|\csg{S}{i}|)\rfloor}.
    \end{equation*}  
    $\hspcexp(S)$ is a CDG with the same number of rows of $S$. Each row is of size
    \begin{equation*}
        |\rsg{\hspcexp(S)}{j}| = \min\set{n_A,~\lfloor\gs_A^\dagger(|\rsg{S}{j}|)\rfloor}.
    \end{equation*}  
    We define $\spcexp(S)$ as the CDG
    \begin{equation*}
        \spcexp(S) = \hspcexp(\vspcexp(S)).
    \end{equation*}
\end{defn}

These expansions depend on the rank expansion of $\bA$ and $\bB$, $\gs_A$ and $\gs_B$ (Definition~\ref{def:RankExpDef}).
When $\gs_A$ and $\gs_B$ are understood from the context, we often abbreviate this dependency.
$\vspcexp$ grows $S$ vertically as much as possible such that $\vspcexp(S)$ can have the same span as $S$. Similarly, $\hspcexp$ grows horizontally. 
See steps 3 and 4 in Figure~\ref{fig:gridGrowth} for an illustration.

\begin{figure}[h]
    \includegraphics[width=\linewidth]{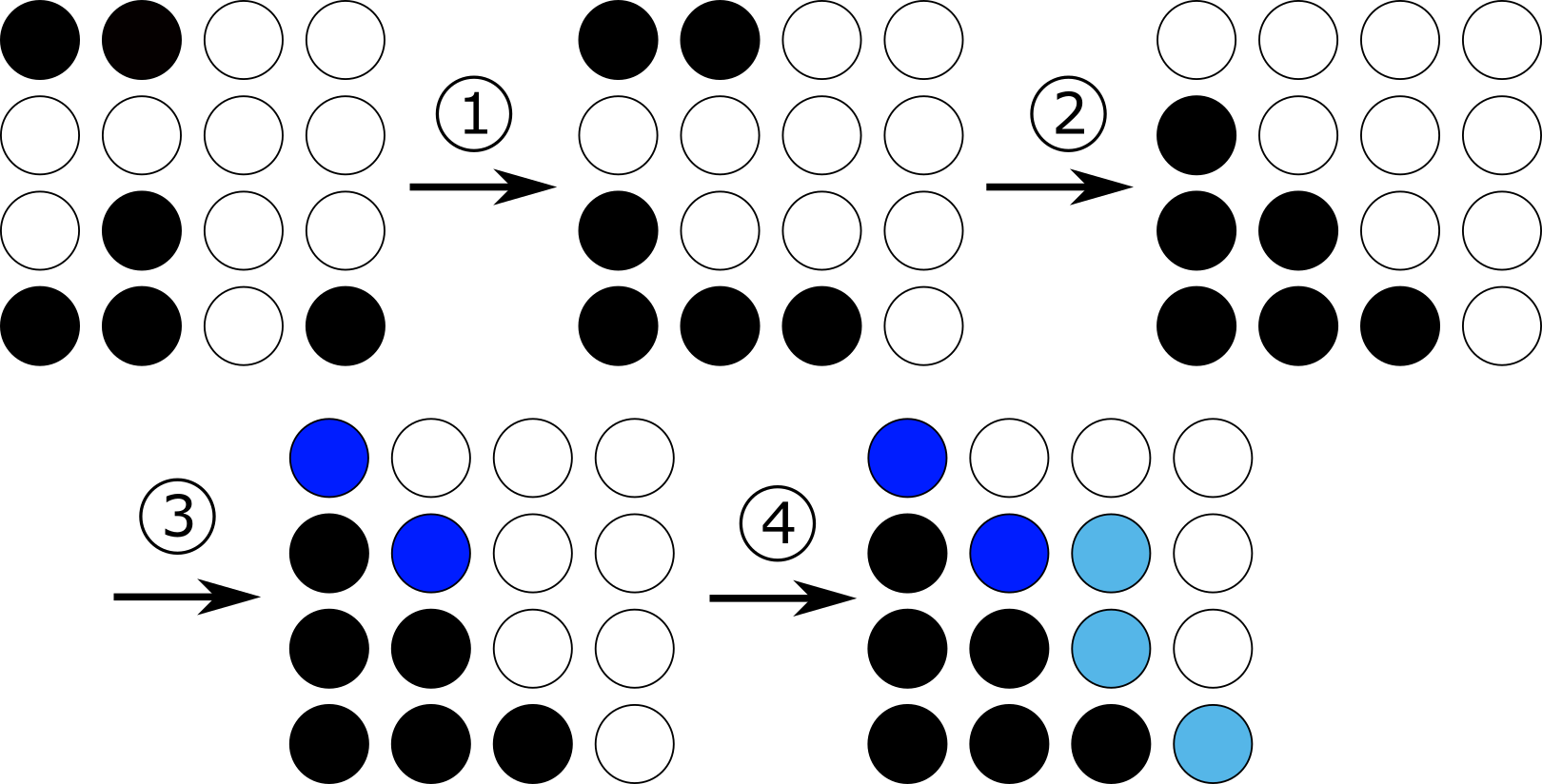}
    \centering
    \caption{Grid expansion of a pre-CDG (Definition~\ref{def:precompact_CDG}) $S$ on a $4 \times 4$ grid where $\gs_A^\dagger=\gs_B^\dagger = f$, and
    $\lfloor f(1) \rfloor=1$, $\lfloor f(2) \rfloor=3$, 
    $\lfloor f(3) \rfloor =4$, and $\lfloor f(4) \rfloor =5$. Through a column (Step 1, denoted by a 1 with a circle) and row (Step 2) reordering, we produce a CDG. 
    We then apply a vertical expansion (Step 3), resulting in two additional points in dark blue. 
    We finish with a horizontal expansion (Step 4), resulting in three more points in light blue.}
    \label{fig:gridGrowth}
\end{figure}

Note that $\spcexp(S)$ only depends on the \textit{shape} of $S$, and it is independent of the column vectors represented by $S$. 
Therefore, we can compute $\spcexp$ for a pre-compact dense grid as defined below. 

\begin{defn} \label{def:precompact_CDG}
    A grid $S$ is a \textit{pre-compact dense grid} (pre-CDG) if $S$ becomes a CDG after a column and row reordering of the grid.
    Denote this CDG by $\cdg(S)$. Then we extend the definition of $\spcexp$ to a pre-CDG $S$ by
    \begin{equation*}
        \spcexp(S) := \spcexp(\cdg(S)).
    \end{equation*}
\end{defn}
This is well-defined since the shape of $\cdg(S)$ is unique. See Figure~\ref{fig:gridGrowth} for an example of converting a pre-CDG to a CDG.

\subsubsection{Bounding General Grid Rank from a Compact Basis}

Now we are ready to prove the main result of the discrete step.

\begin{lemma}\label{lem:discrete_step}
    For any grid $G \sset [n] \times [m]$, there exists a pre-CDG $S$ such that
    \begin{enumerate}[nosep, label=(\roman*)]
        \item $|S| \leq \rank(G)$,
        \item $|G| \leq |\spcexp(S)|$.
    \end{enumerate}
\end{lemma}

\begin{proof}
    Since $\rank(G)$, $|G|$, and the pre-CDG structure are invariant under column permutations on the grid, we may reorder the columns and assume $D = \vcl(G)$ is a CDG.

    \textit{Step 1. Construction of pre-CDG $S$}. 
    Let $B_G$ and $B_D$ be the bases of $G$ and $D$, respectively.
    Construct grid $S \sset B_D$ in the following way.
    Let $X = P_A(B_D)$. 
    For each column index $p \in X$,
    we remove grid points in
    ${B_D}_{[p, \cdot]}$ from the top until there are at most $r_p = \rank(G_{[p,\cdot]})$ points left (we may not need to remove anything). 
    Repeat this for each column, and let $S$ be the resulting subgrid of $B_D$.
    Clearly, $S$ is a pre-CDG since $B_D$ is a pre-CDG by Proposition \ref{prop:tnsprodB}. 
    Removing the top points in each column preserves the pre-CDG structure.

    \textit{Step 2. Proof of (i)}.
    Let $B_G^k$ be the basis for $G \cap([k] \times [m])$.      
    With increasing number of columns $k$, we have
    \begin{equation*}
        B_G^1 \sset B_G^2 \sset \cdots \sset B_G^n = B_G.
    \end{equation*}
    Similarly, define $S^k = S \cap ([k] \times [m])$. 
    We show that $|S^k| \leq |B_G^k|$ for all $k$ by an induction on $k$. As defined in step 1, let $r_p = \rank(\csg{G}{p})$.

    The base case $|S^1| \leq r_1 = |B_G^1|$ is immediate. Now suppose $|S^i| \leq |B_G^i|$ for all $i \leq k$. We wish to show $|S^{k+1}| \leq |B_G^{k+1}|$. 
    Indeed, if $\ba_{k+1} \in \spp\set{\ba_i}_{i \leq k}$, then by Proposition \ref{prop:tnsprodB}, $|S^{k+1}| = |S^{k}|$, so $|S^{k+1}| = |S^{k}| \leq |B_G^k| \leq |B_G^{k+1}|$. 
    If $\ba_{k+1} \notin \spp\set{\ba_i}_{i \leq k}$, we have $|S^{k+1}| - |S^{k}| \leq r_{k+1}$. It is sufficient to show $|B_G^{k + 1}| - |B_G^{k}| \geq r_{k+1}$.
    To that end, when $\ba_{k+1} \notin \spp\set{\ba_i}_{i \leq k}$, by Lemma \ref{lem:newcol}, we must add $(k+1, q)$ to the basis {(because of Algorithm~\ref{alg:basis_construction})} if $\bb_q \notin \spp\set{\bb_j: 
    {\exists (p, j) \in G~\text{s.t.}~j < q}}$. 
    Thus, we have to add at least a maximal linearly independent set of vectors $\csg{G}{k+1}$ to the basis (but may include more).
    Hence, $|B_G^{k + 1}| - |B_G^{k}| \geq r_{k+1}$.
    Combining the two cases, we conclude $|S^k| \leq |B_G^k|$ for all $k$. 
    In particular, $|S| = |S^n| \leq |B_G^n| = |B_G|$ as desired.

    \textit{Step 3. Proof of (ii)}.
    Let $R_c$ and $R_r$ be respectively the column and row reordering after which $S$ becomes a CDG. That is, $\cdg(S) = R_r(R_c(S))$.
    Let $S^+ = \vspcexp(\cdg(S))$.
    First, we show that for each $i \in P_A(R_c(B_D))$, $R_c(D)_{[i, \cdot]} \sset S^+_{[i, \cdot]}$. 

    Fix a column $i \in P_A(R_c(B_D))$. Denote $i'$ as the index of this column before applying $R_c$. 
    Denote 
    \begin{align*}
        s &= |\csg{\cdg(S)}{i}| = |\csg{S}{i'}|,\\
        g &= |\csg{R_c(G)}{i}| = |\csg{G}{i'}|,\\
        d &= |\csg{R_c(D)}{i}| = |\csg{D}{i'}|.
    \end{align*}
    Since $S^+$ and $R_c(D)$ are both dense grids, it suffices to check $d \leq |S^+_{[i, \cdot]}|$.
    Clearly, $g = d$. 
    If no point was removed from $B_D$ to form $S$ in column $i'$, then 
    by part (iv) of Proposition \ref{prop:tnsprodB} and the fact $\gs_B$ is a lower bound on the rank expansion for $\bB$, $d \leq |S^+_{[i, \cdot]}|$.
    If grid points were removed from $B_D$ in column $i'$, then $s = r_{i'}$.
    Now $|S^+_{[i, \cdot]}| = \lfloor\gs_B^{\dagger}(s)\rfloor = \lfloor\gs_B^{\dagger}(r_{i'})\rfloor \geq g = d$. 
    Thus, in both cases, $R_c(D)_{[i, \cdot]} \sset S^+_{[i, \cdot]}$.

    The above shows for all $i \in P_A(R_c(B_D))$, The $i$th column $\csg{R_c(D)}{i}$ is covered by $S^+_{[i, \cdot]}$. 
    In particular, for each $j \in P_B(D)$, $\rsg{R_c(B_D)}{j}$ is covered by $S^+_{[\cdot, j]}$.
    By part (v) of Proposition \ref{prop:tnsprodB} and the fact that $\gs_A$ is a rank expansion lower bound for $\bA$,
    we have $|\hspcexp(S^+)_{[\cdot, j]}| \geq |D_{[\cdot, j]}|$. 
    Since this holds for every $j \in P_B(D)$, we conclude $|\spcexp(S)| \geq |D| = |G|$ as desired.
\end{proof}

With this established, it remains to find a strictly increasing function $\phi$ such that for any CDG $S$, $|\spcexp(S)| \leq \phi(|S|)$.
Together with Lemma \ref{lem:discrete_step}, this will imply for any $G$, there exists some pre-CDG $S$,
\begin{equation*}
    \phi(\rank(G)) \geq \phi(|S|) = \phi(|\cdg(S)|) \geq |\spcexp(\cdg(S))| \geq |G|.
\end{equation*}
Thus, $\gs_C = \phi^{-1}$ is a valid rank expansion.
We find such a function $\phi$ in the next section.

\section{Rank Expansion Bounds via Continuous Analysis of Compact Grid Expansion} \label{sec:cont_step}

Finding a tight upper bound on $|\spcexp(S)|$ (given $\gs_A$ and $\gs_B$) for all CDG $S$ with a given size can be a hard discrete optimization problem for arbitrary $\gs_A$, $\gs_B$, and CDG $S$.
In order to find an easy-to-apply bound, we relax the integer-grid assumption and work with the so-called ``stairs" in $\R^2$.
This is made precise in Section \ref{sec:stair} below.
As required by Theorem \ref{thm:lb} and Theorem \ref{thm:nest_lb}, we assume 
\begin{equation} \label{eq:assumpt_zero}
    \gs_A(0) = \gs_B(0) = 0 \text{~and they are concave.}
\end{equation} 
Throughout this section, we assume in addition 
\begin{equation} \label{eq:assumpt_smooth}
    \gs_A \text{~and~} \gs_B \text{~are strictly increasing \editstwo{$C^1$} functions.} 
\end{equation}
This additional assumption \eqref{eq:assumpt_smooth} will not make the resulting bound $\gs_C$ worse through a density argument on these functions.
\edits{Recall the pseudoinverse function denoted with a dagger, as defined in~\eqref{eq:defpseudoinv}}. Now we may simplify our notation and denote
\begin{equation}
    \begin{aligned} \label{eq:def_fg}
        f(x) &:= \gs_A^{-1}(x) = \gs_A^\dagger(x),\\
        g(x) &:= \gs_B^{-1}(x) = \gs_B^\dagger(x).
    \end{aligned}
\end{equation}
These are strictly increasing convex functions with $f(0) = g(0) = 0$.
We will assume these notations in the following context unless otherwise explained.
To simplify notations, 
for arbitrary functions $h_1$ and $h_2$, we use the notation 
\begin{equation} \label{eq:propto_def}
  \edits{h_1(x)\propto_+ h_2(x)}
\end{equation}
\edits{if there exists $c > 0$ such that $h_1(x) = c \cdot h_2(x)$}.

\subsection{The Stair Relaxation} \label{sec:stair}

We identify a CDG with a stair-like area in $\mathbb{R}^2$ -- each grid point $(i, j)$ now represents a unit square $[i-1, i] \times [j-1, j]$ in $\R^2$.
An illustration can be found in Figure~\ref{fig:merge}.
This \textit{stair} structure consists of contiguous rectangles with decreasing heights and varying widths. We will refer to these rectangles as \textit{\textit{steps}}. 
The colloquial name of stairs represents the visual similarity to the stairs we see in real life in multi-story buildings. 
Now we define the continuous version of $\spcexp$ on stairs.

\begin{defn}  \label{def:langleS}
    The \textit{expansion size} of a CDG $S$ (with respect to $f$ and $g$) is defined as
    \begin{equation}\label{eq:expsize}
        \ab{S} = \int_{S \sset \R^2} df(x) dg(y).
    \end{equation}
    We identify the grid $S$ with a stair in $\R^2$ in the integration.
\end{defn}

The expansion size $\ab{S}$ now serves as an upper bound of $|\spcexp(S)|$, as stated below.

\begin{lemma} \label{lem:gridexp_ab}
    \sloppy For any CDG $S$, we have
    $|\spcexp(S)| \leq \ab{S}$.
\end{lemma}

\begin{proof}
    Let $S^+ = \vspcexp(S)$ and $S^{++} = \hspcexp(S^+)$.
    By definition of these expansions,
    \begin{align*}
        |\csg{\vspcexp(S)}{i}| 
        &= 
        \min\set{n_B,~\lfloor\gs_B^\dagger(|\csg{S}{i}|)\rfloor}
        \leq
        g(|\csg{S}{i}|),\\
        |\hspcexp(S^+)_{[\cdot, j]}| 
        &= 
        \min\set{n_A,~\lfloor\gs_A^\dagger(|S^+_{[\cdot, j]}|)\rfloor}
        \leq
        f(|S^+_{[\cdot, j]}|).
    \end{align*} 
    When these are indeed equalities, because of the stair structure in $S^{++}$ and the fact $f(0) = g(0) = 0$, we have
    \begin{equation*}
        |\spcexp(S)| = |S^{++}| = \ab{S}.
    \end{equation*}
    Therefore, the conclusion follows.
\end{proof}

We are left to bound $\ab{S}$ given $|S| = \int_{S} dxdy$.
To that end,
we further relax the corner points of $S$ from being integer coordinates to general points in (the positive quadrant of) $\mathbb{R}^2$.
The integral version of definitions of $\ab{S}$ and $|S|$ extend naturally to a non-integral stair.
The decreasing stair structure is still preserved,
i.e., $S$ is the region below a decreasing step function in the first quadrant.
In the next section, we prove that $\ab{S}$ is maximized when $S$ is a single rectangle when its size $|S|$ is fixed.

\subsection{Reduce Stairs to Rectangles}

First, we introduce the \textit{merge} operation that reduces the number of steps (rectangles) in the stair by 1.

\begin{defn} \label{def:merge_step}
    Let $S$ be a stair with at least 2 steps indexed by $1, 2, ...$ (from left to right). 
    Fix two consecutive steps $k$ and $k+1$.
    Denote $u$ as the difference in their heights.
    There are two ways of changing the shape of $S$ with $|S|$ and its stair structure unchanged:
    \begin{enumerate}[nosep, label=(\roman*)]
        \item increase $u$: until step $k$ has the same height of stair $k-1$ (if $k > 1$) or step $k+1$ has the same height of stair $k+2$ (or down to the ground if it is the last step).
        \item decrease $u$: until step $k$ and $k+1$ reaches the same height ($u = 0$).
    \end{enumerate}
    Denote the resulting stair from the first approach as $M_1$ and the second approach as $M_2$.
    The \textit{merge} operation (on a given stair $S$) is defined as
    \begin{equation*}
        \merge(k, k+1) = 
        \begin{cases}
            M_1 & \text{if~} \ab{M_1} \geq \ab{M_2}\\
            M_2 & \text{if~} \ab{M_1} < \ab{M_2}
        \end{cases}.
    \end{equation*}
\end{defn}

An illustration of the merge step is given in Figure \ref{fig:merge} below.
The next step is to show that a merge always increases the expansion size.

\begin{figure}[htbp] 
    \centering
    \includegraphics[width = 0.45\linewidth]{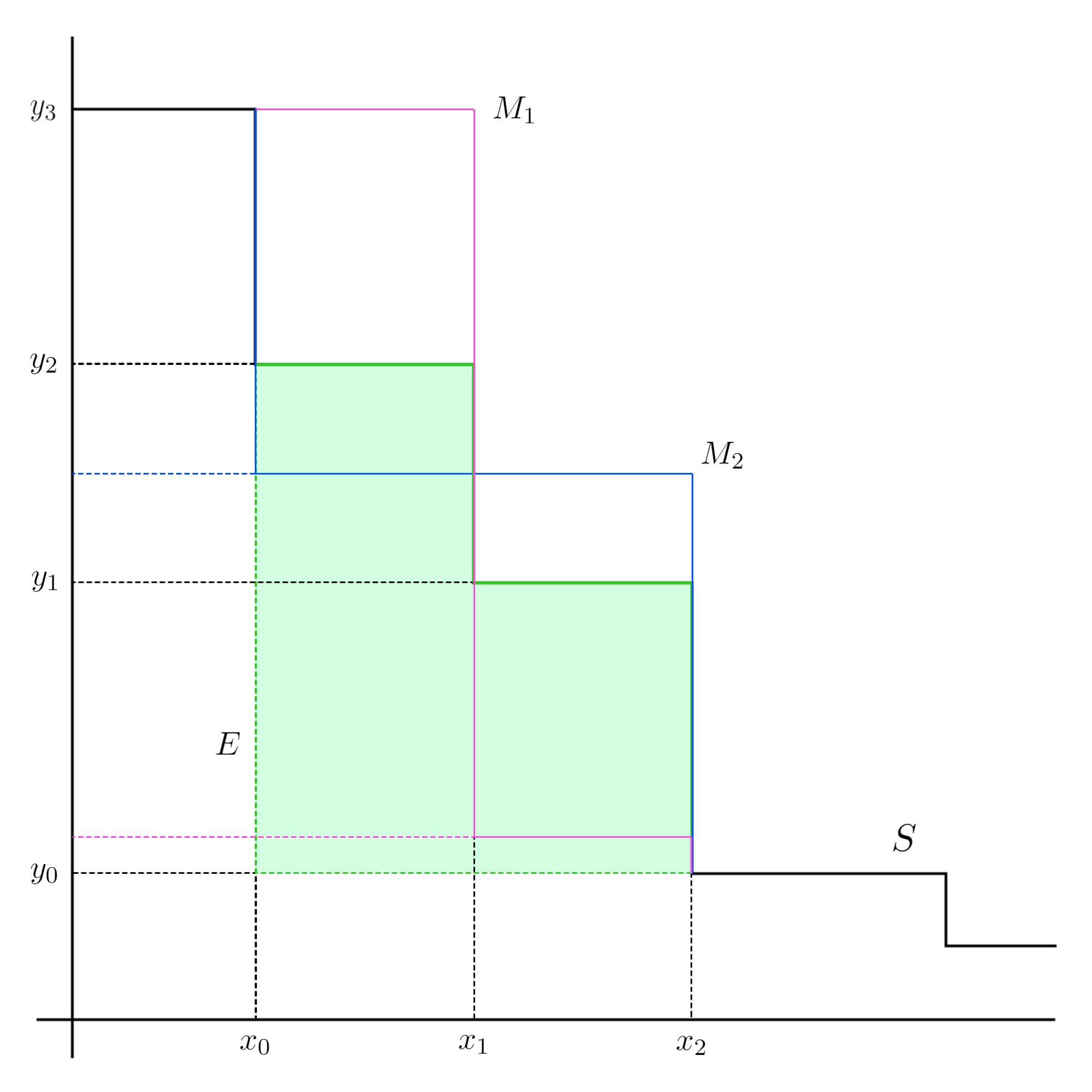}
    \caption{
    Stair $S$ consists of four steps in decreasing height, and we want to ``remove'' a step without changing the area $|S|$. 
    Consider the operation $\merge(2, 3)$.
    The original shape of $S$ (relevant to this merge) is highlighted in green. 
    The pink stair represents $M_1$, produced by increasing $u$ to the maximum.
    The blue stair represents $M_2$, produced by decreasing $u$ to 0.
    }
    \label{fig:merge}
\end{figure}

\begin{lemma} \label{lem:continuous_step}
    Let $S$ be a CDG on the grid with at least two steps.
    Identify $S$ as a stair.
    Consider the merge of steps $k$ and $k+1$. 
    The merge operation varies their height difference while keeping the area $|S|$ unchanged.
    Denote $E(u)$ as the (unique) stair when their height difference is $u$. 
    Then for any $[a, b] \sset [0, +\infty)$, 
    \begin{equation*}
        \max_{u \in [a, b]} \ab{E(u)} = \max\set{\ab{E(a)}, \ab{E(b)}}.
    \end{equation*}
    Hence,
    for any $k$ and $u$ in the merge,
    \begin{equation*}
        \ab{E(u)} \leq \ab{\merge(k, k+1)} = \max\set{\ab{M_1}, \ab{M_2}}.
    \end{equation*}
    Consequently, by carrying out a sequence of $\merge$, for any CDG $S$ with size $|S| = t$,
    \begin{equation} \label{eq:opt_phi_exp}
        \ab{S} \leq \max_{\substack{t_A,~ t_B \geq 1,\\ t_A t_B = t}} f(t_A)g(t_B).
    \end{equation} 
\end{lemma}

\begin{proof}
    We start by introducing some notations.
    Consider the stair $E(u)$.
    As depicted in Figure \ref{fig:merge}, let $x_0 < x_1 < x_2$ be the horizontal splits between steps $(k-1, k)$ (or $x_0 = 0$ if $k = 1$), $(k, k+1)$, and $(k+1, k+2)$ (or ground if $k+1$ is the last step).
    Let $y_0 < y_1 < y_2 < y_3$ be the height of steps $k-1, ~k, ~k+1, ~k+2$ (we set $y_0 = 0$ if $k+1$ is the last step, and $y_3 = +\infty$ if $k = 1$).  
    Since $y_2 - y_1 = u$, using the equi-area relation, we have
    \begin{equation}\label{eq.yu_rel}
        \begin{cases}
        y_1 = \bar{y} - \frac{x_1 - x_0}{x_2 - x_0} u,\\
        y_2 = \bar{y} + \frac{x_2 - x_1}{x_2 - x_0} u,
        \end{cases} 
    \end{equation} 
    where $\bar{y}$ is the height of the new $k$th stair in $M_2$ when $u = 0$, which is independent of $u$.

    The expansion size can be computed to be
    \begin{equation*}
        \ab{E(u)} = C + [f(x_1) - f(x_0)] [g(y_2(u)) - g(y_0)] + [f(x_2) - f(x_1)] [g(y_1(u)) - g(y_0)],
    \end{equation*}
    where $C$ is independent of $u$, and other terms represent the area of the green region in Figure \ref{fig:merge}.
    Taking the derivative with respect to $u$ and using the relations in \eqref{eq.yu_rel}, we obtain 
    \begin{align*}
        \frac{d}{du}\ab{E(u)} 
        &= 
        \frac{x_2 - x_1}{x_2 - x_0}[f(x_1) - f(x_0)] g'(y_2) - \frac{x_1 - x_0}{x_2 - x_0} [f(x_2) - f(x_1)] g'(y_1) \\
        &\propto_+ 
        \frac{f(x_1) - f(x_0)}{x_1 - x_0} g'(y_2) - \frac{f(x_2) - f(x_1)}{x_2 - x_1}  g'(y_1)\\
        &=: 
        k_1 g'(y_2) - k_2 g'(y_1).
    \end{align*}
    Applying the relation \eqref{eq.yu_rel} once more, we have 
    \begin{equation*}
        \frac{d}{du}\ab{E(u)} = 0 \ \ \ \ \Leftrightarrow\ \ \ \ 
        g'\left(\bar{y} - \frac{x_1 - x_0}{x_2 - x_0} u\right) = \frac{k_1}{k_2} \cdot g'\left(\bar{y} + \frac{x_2 - x_1}{x_2 - x_0} u\right).
    \end{equation*}
    By convexity of $g$, the right-hand side is increasing in $u$ whereas the left-hand side is decreasing in $u$. 
    Thus, there is at most one critical point for $\ab{E(u)}$ on $u \geq 0$. 
    Now we show that 
    \begin{equation}
        \frac{d}{du}\ab{E(u)}|_{u = 0} \leq 0. 
    \end{equation}
    Indeed, when $u = 0$, $y_1 = y_2 = \bar{y}$, and by convexity of $f$, $k_1 \leq k_2$ by the slopes of secant lines of $f$.
    Therefore, 
    \begin{equation*}
        \frac{d}{du}\ab{E(u)}|_{u = 0} = (k_1 - k_2)g'(\bar{y}) \leq 0. 
    \end{equation*}
    Thus, $\ab{E(u)}$ cannot have a local maximum in any $(a, b) \sset \R_+$. 
    Applying this to the case of $\merge(k, k+1)$, the maximum of $\ab{E(u)}$ is either $\ab{M_1}$ or $\ab{M_2}$. 

    We repeatedly apply $\merge(2, 3)$ until the grid has only two steps (if there are only 1 or 2 steps to begin with, we skip this step). 
    Then we apply $\merge(1, 2)$, yielding a $t_A \times t_B$ rectangle, with size $t = t_A \times t_B$. 
    This rectangle has a width and height of at least 1. 
    This is because we start from a CDG $S$ on the grid, so the first and last step of $S$ has a height and width of at least 1.
    Hence, in the last merge, the two steps are of height and width of at least 1.
    Hence, the optimization problem~\eqref{eq:opt_phi_exp} is a valid upper bound of $\ab{S}$, which maximizes the grid expansion over all rectangular-shaped grids $S$ with an area at least $t$ and width and height at least 1. 
\end{proof}

As motivated at the end of Section \ref{sec:main_proof},
Lemma \ref{lem:gridexp_ab} and Lemma \ref{lem:continuous_step} give us a suitable function $\phi$ such that for any CDG $S$, $|\spcexp(S)| \leq \phi(|S|)$,
\begin{equation} \label{eq:main_phi}
    \phi(t) = \max_{\substack{t_A,~ t_B \geq 1,\\ t_A t_B = t}} f(t_A)g(t_B).
\end{equation}
To conclude the proof of Theorem \ref{thm:lb},
it remains to prove the proposed expansion lower bound
\begin{equation*}
    \gs_C(k) = \min_{\substack{k_A\geq d_A,~k_B \geq d_B\\ 
    k_A k_B \geq k}}\gs_A(k_A)\cdot\gs_B(k_B)
\end{equation*}
is indeed a lower bound of $\phi^{-1}$. 
In fact, in the next section, we prove that it is optimal, in that $\gs_C = \phi^{-1}$.

\subsection{Proof of Main Expansion Bound (Theorem \ref{thm:lb})} \label{sec:main_proof_1}


We now provide the final step toward proving Theorem~\ref{thm:lb}, showing that the inverse of the growth function $\phi$ gives us the rank expansion function, $\gs_C = \phi^{-1}$.
\begin{lemma} \label{lem:invphi_main}
    Suppose $\gs_A$ and $\gs_B$ are concave, strictly increasing, and smooth with $\gs_A(0) = \gs_B(0) = 0$.
    Let $f = \gs_A^{-1}$, $g = \gs_B^{-1}$, $d_A = f(1)$, $d_B = g(1)$. 
    Define
    \begin{align*}
        \gs_C(k) &= \min_{\substack{k_A\geq d_A,~k_B \geq d_B\\ 
        k_A k_B \geq k}}\gs_A(k_A)\cdot\gs_B(k_B) ,\\
        \phi(t) &= \max_{\substack{t_A,~ t_B \geq 1,\\ t_A t_B = t}} f(t_A)g(t_B).
    \end{align*}
    Then $\gs_C = \phi^{-1}$ on $[d_Ad_B, +\infty)$.
\end{lemma}

\begin{proof}
    To begin with, since $f$ and $g$ are strictly increasing, so is $\phi$, and thus $\phi$ is invertible. Also, $\gs_C$ is strictly increasing on $[d_Ad_B, +\infty)$ because both $\gs_A$ and $\gs_B$ are strictly increasing.

    \textit{Step 1. $\gs_C \leq \phi^{-1}$.}
    Assume by contradiction there is some $k \geq d_Ad_B$ where 
    $\sigma_C(k) > \phi^{-1}(k)$. 
    Then there exists a $k'$ such that 
    \begin{equation*}
        k' := \phi(\gs_C(k)) > \phi(\phi^{-1}(k)) = k.
    \end{equation*}  
    From the above inequality, this means that there exists $t_A, t_B \geq 1$ such that $t_At_B = \gs_C(k)$ and $\gs_A^{-1}(t_A) \gs_B^{-1}(t_B) \geq k'$. 
    This implies
    \begin{align*}
        \gs_C(k) &< \gs_C(k')\\
        &=
        \min_{\substack{k_A\geq d_A,~k_B \geq d_B\\ 
        k_A k_B \geq k'}}\gs_A(k_A)\cdot\gs_B(k_B) \\
        &\leq 
        \gs_A(\gs_A^{-1}(t_A)) \cdot \gs_B(\gs_B^{-1}(t_B)) \\
        &=
        \gs_C(k),
    \end{align*}
    which contradicts that $\gs_C$ is strictly increasing.

    \textit{Step 2. $\phi^{-1} \leq \gs_C$.}
    We show that if $\gs_C(k) = t$ with $k \geq d_Ad_B$, then $\phi(t) \geq k$.
    Indeed, when $k \geq d_Ad_B$, $\gs_C(k) = t$ implies
    \begin{equation*}
        \gs_C(k) = \min_{\substack{k_A\geq d_A,~k_B \geq d_B\\ 
        k_A k_B = k}}\gs_A(k_A)\cdot\gs_B(k_B) = t.
    \end{equation*}
    Note the equality constraint under the $\min$.
    Let $t_A = \gs_A(k_A) \geq 1$ and $t_B = \gs_B(k_B) \geq 1$.
    Then we have found
    \begin{equation*}
    t_A \geq 1,~t_B \geq 1,~t_At_B = t, \text{~such that~}
    \gs_A^{-1}(t_A)\gs_B^{-1}(t_B) = k. 
    \end{equation*}
    By definition of $f$, $g$, and $\phi$, we have $\phi(t) \geq k$, and the proof is complete.
\end{proof}

Finally, we put all the pieces together to give the full proof of Theorem \ref{thm:lb}. 
For convenience, we restate the theorem here.
\edits{
\mainthmone*
}

\begin{proof}
    First, for $k \leq d_Ad_B$, the bound is trivial since
    \[
        \gs_C(k) = \gs_A(d_A)\cdot\gs_B(d_B) = 1,
    \]
    which is clearly a lower bound of the rank of any nonzero matrix. 

    Let us consider $k \geq d_Ad_B$. 
    Through a density argument, we may assume $\gs_A$ and $\gs_B$ are strictly increasing and smooth.
    For any grid $G$ of size $k$, by Lemma \ref{lem:discrete_step}, we can find a pre-CDG $S$ such that 
    \begin{align*}
        |S| &\leq \rank(G),\\
        k &= |G| \leq |\spcexp(S)|.
    \end{align*}
    By Lemma \ref{lem:gridexp_ab} and Lemma \ref{lem:continuous_step}, we have
    \begin{align*}
        |\spcexp(S)| &= |\spcexp(\cdg(S))| 
        \leq \ab{\cdg(S)},\\
        \ab{\cdg(S)}
        &\leq 
        \max_{\substack{t_A,~ t_B \geq 1,\\ t_A t_B = |S|}} f(t_A)g(t_B) =: \phi(|S|),
    \end{align*}
    where $f = \gs_A^{-1}$, $g = \gs_B^{-1}$.
    Finally, combining all the steps, since $\phi$ is increasing,
    \begin{equation*}
        |G| \leq \phi(|S|) \leq \phi(\rank(G)).
    \end{equation*}
    Therefore, Lemma \ref{lem:invphi_main} tells us
    \begin{equation*}
        \rank(G) \geq \phi^{-1}(|G|) = \gs_C(|G|).
    \end{equation*}
    The proof is then complete.
\end{proof}

\subsection{Proof of Nested Expansion Bound (Theorem \ref{thm:nest_lb})} \label{sec:main_proof_2}


We now seek to extend the nested rank expansion lower bound (Theorem \ref{thm:lb}), to cases when $\bC = \otimes_{i=1}^p \bA_i$ contains $p>2$ terms.
Theorem \ref{thm:lb} may not be \editstwo{easy to apply} repeatedly, since the resulting bound $\gs_C$ does not satisfy the assumptions applied to $\gs_A$ and $\gs_B$, as $\gs_C(0) = 0$ and need not be concave.
However, we can circumvent these issues when $\gs_A$ and $\gs_B$ are \textit{log-log concave} (see Definition \ref{def:loglog}).
For such rank expansion functions, we give a simple form for the rank expansion of a $p$-term Kronecker product in Theorem \ref{thm:nest_lb}.
Before going into its proof, we summarize some properties of log-log convex (concave) functions. 
For a more detailed reference on log-log convex functions, see \cite{agrawal2019disciplined}.

\begin{proposition} \label{prop:lccv_x}
    $f(x)$ is log-log convex (resp. concave) if and only if $\ln f(e^x)$ is convex (resp. concave) in $x$. Thus, if $\set{f_i(x)}_{i \in \cI}$ and $\set{g_i(x)}_{i \in \cI}$ are, respectively, two collections of log-log convex and concave functions, then $\sup_i f_i$ and $\prod_i f_i$ are log-log convex functions, and $\inf_i g_i$ and $\prod_i g_i$ are log-log concave functions.
\end{proposition}

\begin{proposition} \label{prop:lccv_inv}
    If $f : \R_+ \mapsto \R_+$ is invertible and log-log convex (resp. concave), then $f^{-1}$ is log-log concave (resp. convex).
\end{proposition}


\begin{proposition} \label{prop:lccv_apprx}
    Let $f$ be a log-log convex (resp. concave) function on $[a, b]$. Then for any $\gee > 0$, there exists a smooth log-log convex (resp. concave) function $g$ on $[a, b]$ with $\|f - g\|_{L^\infty(a, b)} < \gee$.   
\end{proposition}

This function class contains many useful functions for our work.
We provide 
a couple of increasing, concave, and log-log concave functions below. 
\begin{proposition} \label{prop:monom_log}
    For the following choices of $f$, $f(x)$ is concave and log-log concave, and for any $t \in \R_+$, $f(t) - f(t-x)$ is log-log convex on $(0, t)$:
    \begin{enumerate}[nosep, label=(\roman*)]
        \item $f(x) = a \ln (bx + 1)$, $a > 0,~b > 0$;
        \item $f(x) = ax^p$, $a > 0, ~p \leq 1$.
    \end{enumerate}
\end{proposition}

For a more comprehensive list of functions, see again~\cite{agrawal2019disciplined}. 
Note that not all positive increasing log-log concave functions are concave, for example, $f(x) = x \ln (1 + x)$ is convex but log-log concave by Proposition \ref{prop:lccv_x} and Proposition \ref{prop:monom_log}.
The log-log convexity of $f(t) - f(t - x)$ will be useful later in Appendix~\ref{sec:simplify_L_shape} to simplify optimization problems.

Assuming that $\gs_A$ and $\gs_B$ are log-log concave, we can further reduce the minimization problem in Theorem~\ref{thm:lb}. 

\begin{lemma} \label{lem:bndoptim}
    If $f$ and $g$ are positive increasing log-log concave functions on $\R_+$, then for any positive numbers $a, b, c, d, k$ with $b$ and $d$ being possibly infinite,
    \begin{equation*}
        \min_{\substack{k_A\in [a, b],~k_B \in [c, d]\\ 
        k_A k_B \geq k}}f(k_A)\cdot g(k_B)
        =
        \min_{\substack{k_A\in [a, b],~k_B \in [c, d]\\ k_A\in \set{a, b} \text{ or } k_B \in \set{c, d}, \\
        k_A k_B \geq k}}f(k_A)\cdot g(k_B).
    \end{equation*}
\end{lemma}
\begin{proof}
    By monotonicity, when $k \leq ac$, the minimum is $f(a)g(c)$, so we are done.
    When $k \geq ac$, fixing $k$, the optimization is equivalent to
    \begin{equation*}
        \min_{\substack{k_A \in [a, b],~k_B \in [c, d],\\ k_Ak_B = k}} f(k_A) g(k_B) = \min_{\substack{x \in [a, b],\\k/x \in[c. d]}} f(x) g(k/x).
    \end{equation*}
    By monotonicity, for a finite $k$, we can always assume $b$ and $d$ are finite large numbers. 
    Through a density argument using Proposition \ref{prop:lccv_apprx}, we can further assume that $f$ and $g$ are smooth and strictly increasing on the intervals $[a, b]$ and $[c, d]$, respectively.
    Let $h(x) = f(x) g(k/x)$, then  
    \begin{equation*}
        h'(x) = \frac{1}{x}f(x)g\left(k/x\right)\left[\frac{xf'(x)}{f(x)} - \frac{(k/x)g'(k/x)}{g(k/x)}\right].
    \end{equation*}
    When $f$ and $g$ are log-log concave, by Proposition \ref{prop:lccv_x}, $F(x) = \ln f(e^x)$ and $G(x) = \ln g(e^x)$ are concave. Note that
    \begin{equation*}
        \frac{xf'(x)}{f(x)} = F'(\ln x),\ \ \ \ 
        \frac{(k/x)g'(k/x)}{g(k/x)} = G'(\ln k - \ln x). 
    \end{equation*}
    Hence $h'(x) \propto_+ F'(\ln x) - G'(\ln k - \ln x)$ (\edits{$\propto_+$ means positively proportional to, see \eqref{eq:propto_def}}). Since $F'$ and $G'$ are decreasing, $h(x)$ cannot have a local minima. Thus, the minimum of $h$ is attained on the boundary.
\end{proof}

With this simplification, we can prove Theorem \ref{thm:nest_lb}. 
Again, for convenience, we restate it here.
\edits{
\mainthmtwo*
}

\begin{proof}
    Equation \eqref{eq:lb_for_many} will follow via a proof by induction on $p$. 
    In the base case $p = 2$, we have by Theorem \ref{thm:lb} and Lemma \ref{lem:bndoptim}, 
    \begin{equation*}
        \gs_C(k) = \min\set{\gs_1(k / d_2),~\gs_2(k / d_1)}
    \end{equation*}
    is a valid concave and log-log concave rank expansion lower bound for $\bC$.
    For $p > 2$, let $\bB = \bigotimes_{i < p} \bA_i$.
    By the base case and the induction hypothesis,
    \begin{align*}
        \gs_C(k) &= \min\set{\gs_B(k / d_p),~\gs_p(k / d_B)} \\
        &=
        \min\set{\min_{j < p}\set{\gs_j\left(\frac{k / d_p}{\prod_{i \neq j, i < p}d_i}\right)}, \gs_p\left(\frac{k}{\prod_{i < p}d_i}\right)} \\
        &= \min_{j\leq p}\set{\gs_j\left(\frac{k}{\prod_{i \neq j, i \leq p}d_i}\right)},
    \end{align*}
    where we used the fact that $d_B = \gs_B^\dagger(1)= \prod_{i < p} d_i$, which is readily verified by an induction. Thus \eqref{eq:lb_for_many} is established, and it is clear that $\gs_C$ is again concave and log-log concave.
    
    If $\gs_i(k) = (k / k_i)^{q_i}$, then $d_i = k_i$, and
    \begin{equation*}
        \gs_j\left(\frac{k}{\prod_{i \neq j}d_i}\right)
        =
        \left(\frac{k / d_j}{\prod_{i \neq j}d_i}\right)^{q_j}
        =
        \left(\frac{k}{\prod_i k_i}\right)^{q_j}.
    \end{equation*}
    The result follows after a trivial minimization.
    
    For $\gs_i(k) = a_i \ln (b_i k + 1)$, we proceed by induction on $p$. We assume that $a_1 \geq \ldots \geq a_{p-1} \geq a_p$, and again let $\bB = \bigotimes_{i < p} \bA_i$.
    Then
    \begin{align*}
        \gs_C(k) &=
        \min\set{\gs_B(k/d_p), \gs_p(k/d_B)} \\
        &=
        \min\set{a_{p-1} \ln \left(b_B\frac{k}{d_p} + 1\right),~
                    a_{p} \ln \left(b_p\frac{k}{d_B} + 1\right)},
    \end{align*}
    where $d_p = \gs_p^{\dagger}(1) = b_{p}^{-1} (e^{1/a_{p}} - 1)$ and $d_B = \gs_B^{\dagger}(1) = b_B^{-1} (e^{1/a_{p-1}} - 1)$, and so by the induction hypothesis, $b_B = \displaystyle \frac{e^{1/a_{p-1}}-1}{\prod_{i < p}b_i^{-1}(e^{1/a_i}-1)}$.
    
    When $ k < d_B d_p$, \edits{$\gs_C(k) \leq 1\leq \tilde{\gs}_C(k)$, so $\gs_C(k)$ is a valid rank expansion lower bound for any nonzero matrix.} 
    Consider $k \geq d_B d_p$. Let $k = \gc d_Bd_p$, $\gc \geq 1$, we have
    \begin{equation} \label{eq:qgp_a80}
        \gs_C(k) =
        \min\set{
        a_{p-1} \ln \left((e^{1/a_{p-1}}-1)\gc + 1\right),~a_{p} \ln\left((e^{1/a_{p}}-1)\gc + 1\right)}.
    \end{equation}
    We prove that the second term in the min above is smaller than the first term. With this established, one can directly verify that $\gs_C$ coincides with the proposed function given by \eqref{eq:lb_for_log}, and the proof is then complete.

    \sloppy To that end, consider $g(x) = \frac{h(x)}{x}$, where $h(x) = \ln\left(\gc (e^x - 1) + 1\right)$. 
    The two terms in the minimization of~\eqref{eq:qgp_a80} are, respectively, $g(1/a_{p-1})$ and $g(1/a_{p})$. 
    Now, it suffices to show that $g$ is decreasing on $\R_+$.
    Indeed, note that $h(0) = 0$, and one can check that $h$ is a concave function on $\R_+$ when $\gc \geq 1$. 
    Thus $g(x)$ is the secant slope of $h$ between $x_1 = 0$ and $x_2 = x$, which is decreasing in $x$.
\end{proof}

\section{Applications} \label{sec:Apps}
We apply the rank expansion lower bounds from the previous section to derive communication lower bounds for several fast bilinear algorithms: Strassen's fast matrix-matrix multiplication, fast convolution, and partially symmetric tensor contractions. \edits{We consider both sequential and parallel communication lower bounds, as we defined in Section~\ref{sec:commlb_frmwrk}.} 
\edits{Throughout this section, we will use the set $\mathcal{P}_n^{(k)}$, which recall is the of matrices with $k$ different columns from an identity matrix of size $n$.}

\subsection{Fast Matrix Multiplication} \label{sec:MMSec}
Communication lower bounds for classical matrix multiplication are usually derived by explicitly reasoning about potential partial sums and applying the Loomis-Whitney inequality~\cite{irony2004communication,loomis1949inequality}.
The expansion bound -- which recall from Definition~\ref{def:ExpansionBnd} bounds the 
largest portion of a bilinear one can complete given a subset of data -- for classical matrix
multiplication \edits{is the function (see, e.g.,~\cite[Lemma
B.1]{solomonik2017communication} or~\cite[Lemma 2.1]{ballard2011minimizing}),
\[
    \mathcal{E}(r^{(A)},r^{(B)},r^{(C)}) = \sqrt{r^{(A)} r^{(B)} r^{(C)}},
\]
where $r^{(A)}$, $r^{(B)}$, and $r^{(C)}$ are the number of elements from the first input, second input, and output stored in cache or on a local machine, respectively. } 
Strassen's algorithm~\cite{strassen1969gaussian}, the most practical known fast algorithm~\cite{ballard2012communication} among those that achieve subcubic complexity, corresponds to a more nontrivial bilinear algorithm than classical matrix multiplication.

We derive communication lower bounds for the bilinear algorithm given by Strassen's approach, which considers any other computational DAG.
These DAGs include ones that do not follow the recursive structure of Strassen's algorithm; they need only compute the same scalar products at the base case level of recursion.
In particular, the operands can be computed by any other additions or linear combinations. 
First, recall that the bilinear algorithm for Strassen's algorithm is as follows.
\begin{defn}[Bilinear Algorithm for Strassen's Matrix Multiplication] \label{def:StrassenBA}
\begin{align*}
    \B{A} = \begin{bmatrix*}[r]
    1&0&1&0&1&-1&0\\
    0&0&0&0&1&0&1\\
    0&1&0&0&0&1&0\\
    1&1&0&1&0&0&-1
    \end{bmatrix*}, \ 
    \B{B} = \begin{bmatrix*}[r]
        1&1&0&-1&0&1&0\\
        0&0&1&0&0&1&0\\
        0&0&0&1&0&0&1\\
        1&0&-1&0&1&0&1
    \end{bmatrix*}, \ 
    \B{C} =  \begin{bmatrix*}[r]
        1&0&0&1&-1&0&1\\
        0&0&1&0&1&0&0\\
        0&1&0&1&0&0&0\\
        1&-1&1&0&0&1&0
    \end{bmatrix*}.
\end{align*}
\end{defn}
In view of these two encoding and decoding matrices, we construct a rank expansion lower bound for the bilinear algorithm $(\B{A},\B{B},\B{C})$. 
\begin{lemma} \label{lem:MMExp}
    \sloppy For a \edits{nested bilinear algorithm} $\big ( \bigotimes_{i=1}^\tau \B{A}, \bigotimes_{i=1}^\tau  \B{B}, \bigotimes_{i=1}^\tau \B{C} \big)$, where $\tau \in \mathbb{Z}_+$ and $(\B{A},\B{B},\B{C})$
    is the bilinear algorithm for Strassen's base algorithm, then the function,
    \[
        \mathcal{E}(r^{(A)},r^{(B)},r^{(C)}) = \min(r^{(A)}, r^{(B)}, r^{(C)})^{\log_2(3)}
    \]
    is an expansion bound function for the nested bilinear algorithm.
\end{lemma}
\begin{proof}
    We first focus on the matrix $\B{A}$ and will later briefly address matrices $\B{B}$ and $\B{C}$. 
    
    By direct calculations (e.g., numerical calculations), the smallest rank for a matrix formed by any subset of $k \in [7]$ different columns from the matrix $\B{A}$ is $1,2,2,3,3,4,4$ for $k=1,\ldots,7$, respectively. Now, define the function
    \[
        \gs(k) = k^{\log_3(2)}.
    \]
    One can directly check the following holds for any $k \in [7]$:
    \[
        \gs(k) \leq \min_{\B{P} \in \mathcal{P}_7^{(k)}} \{ \mathrm{rank}(\B{AP}) \}.
    \]
  \edits{This shows $\gs(k)$ is a rank expansion lower bound for $\B{A}$ (Definition~\ref{def:RankExpLBDef}).} Indeed, $\gs$ is also concave, log-log concave (see Proposition~\ref{prop:monom_log}), and satisfies $\gs(0) = 0$. \edits{Also, recalling the definition of the pseudoinverse from~\eqref{eq:defpseudoinv}}, one can verify $d_A \equiv \gs^\dagger(1) = 1$. This permits us to apply Proposition~\ref{thm:nest_lb} with $q \equiv q_i = \log_3(2) < 1$ to ensure the rank expansion lower bound for $\bigotimes_{i=1}^\tau \B{A}$ for any $\tau \in \mathbb{N}$ and all $k \in [7^\tau]$ is 
    \begin{align*}
        \gs(k) &= k^{q}.
    \end{align*}
    By using a nearly identical calculation, one can confirm $\gs$ is a rank expansion lower bound for $\bigotimes_{i=1}^\tau\B{B}$ and $\bigotimes_{i=1}^\tau \B{C}$ for all $\tau \in \mathbb{N}$. We can then conclude for any $k \in [7^\tau]$ and all $\B{P} \in \mathcal{P}_{7^\tau}^{(k)}$,
    \begin{align*}
        \sigma(k) 
        \leq 
        \min\big\{
        \mathrm{rank}\big([\otimes_{i=1}^\tau \B{A}]\B P\big),
        \mathrm{rank}\big([\otimes_{i=1}^\tau\B{B}]\B P\big),
        \mathrm{rank}\big([\otimes_{i=1}^\tau\B{C}]\B P\big)
        \big\}.
    \end{align*}
    Applying the monotone function $\sigma^{-1}(k) = k^{\log_2(3)}$ to both sides and recalling the choice of the expansion bound function $\mathcal{E}$, we get for all $\B{P} \in \mathcal{P}_{7^\tau}^{(k)}$,
    \begin{align*}
        \#\text{cols}(\B{P}) 
        &\leq \mathcal{E}\Big(
        \mathrm{rank}\big([\otimes_{i=1}^\tau \B{A}]\B P\big),
        \mathrm{rank}\big([\otimes_{i=1}^\tau\B{B}]\B P\big),
        \mathrm{rank}\big([\otimes_{i=1}^\tau\B{C}]\B P\big)
        \Big).
    \end{align*}
    Thus, $\mathcal{E}$ is an expansion bound by Definition~\ref{def:ExpansionBnd} (c.f.~\eqref{eq:rnk_exp_to_exp_bnd}).
\end{proof}
Next, we apply Lemma~\ref{lem:MMExp} to derive a sequential communication lower bound.


\begin{corollary}
\label{cor:MMVLB}
  Given square matrices of size $n \in \mathbb{N}$ (assumed to be a power of 2) and \edits{fast memory of size $M$}, the sequential communication cost of Strassen's fast matrix multiplication algorithm is at least
    \begin{equation*}
      \max \bigg\{
      \frac{2 n^{\log_2(7)}}{M^{\log_2(3)}} \cdot M,
      3n^2 \bigg\}.
  \end{equation*}
\end{corollary}
\begin{proof}
    By Proposition~\ref{prop:edgar_seq_thm} and noting the size of the inputs and output are $n^2$, the sequential communication lower bound of the nested bilinear algorithm is
    \begin{equation} \label{eq:VLB}
        \max \Big\{
            \frac{2RM}{\mathcal{E}^{\max}(M)}, {3n^2}
        \Big \},
    \end{equation}
    where the rank of the nested bilinear algorithm is $R = 7^{\log_2(n)} = n^{\log_2(7)}$. With the help of the expansion bound from Lemma~\ref{lem:MMExp}, the function $\mathcal{E}^{\max}(M)$ (from Proposition~\ref{prop:edgar_seq_thm}) can be expressed as
    \begin{align*}
        \mathcal{E}^{\max}(M) &= \max_{\substack{r^{(A)}, r^{(B)}, r^{(C)} \in \mathbb{N}, \\r^{(A)} + r^{(B)} + r^{(C)} = 3M}} \mathcal{E}(r^{(A)}, r^{(B)}, r^{(B)}) \\
        &= \max_{\substack{r^{(A)}, r^{(B)}, r^{(C)} \in \mathbb{N}, \\r^{(A)} + r^{(B)} + r^{(C)} = 3M}}
        \min(r^{(A)}, r^{(B)}, r^{(C)})^{\log_2(3)} \\
        &=
        M^{\log_2(3)}.
    \end{align*}
    This function is strictly increasing (over the nonnegative reals) and convex, hence, we can substitute it into~\eqref{eq:VLB}. 
\end{proof}

In contrast, the existing sequential communication lower bound for the standard Strassen's algorithm computational DAG~\cite{ballard2013graph} is
 \begin{equation*} \label{eq:MMLBBallard}
      \max \bigg \{ \frac{n^{\log_2(7)}}{M^{\log_4(7)}} \cdot M, 3n^2 \bigg \}.
  \end{equation*}

\begin{remark} \label{rem:rqgp_1}
    \editstwo{We suspect Corollary~\ref{cor:MMVLB}'s lower bound is not tight. The gap may arise because the rank expansion bound of a bilinear algorithm and communication complexity are not tight for certain bilinear algorithms $(\B A, \B B, \B C)$ where only a small subset of columns are low rank in $\B{A}$ (or $\B{B}$ or $\B{C}$) while the remaining are nearly full rank. For example, the subset of columns with indices \{1, 2, 4\} from $\B A$ for Strassen's bilinear algorithm (Definition~\ref{def:StrassenBA}) together have rank 2 while any 3 columns from the remaining 4 columns are full rank. The proof of Proposition~\ref{prop:edgar_seq_thm} from~\cite{solomonik2017communication} performs a worst-case analysis where if there exists a subset of columns that is low rank in, for example $\B{A}$, then the low-rank structure is assumed to persist in any remaining subset of columns in $\B{A}$.}
\end{remark}

Next, we consider lower bounds on the parallel communication cost.
\begin{corollary} 
\label{cor:MMHLB}
    Given square matrices of size $n$ (assumed to be a power of 2) and \edits{$P$ processors,}
  the parallel communication cost of 
  Strassen's fast matrix multiplication algorithm is at least
  \begin{equation*}
      3 \cdot \bigg( \frac{n^{\log_3(7)}}{P^{\log_3(2)}} - \frac{n^2}{P}\bigg).
  \end{equation*}
\end{corollary}
\begin{proof}
    By Proposition~\ref{prop:edgar_par_thm} and noting the size of the inputs and output are $n^2$, the communication lower bound of the nested bilinear algorithm is no
    smaller than the sum of some natural numbers $r^{(A)},r^{(B)},r^{(C)} \in \mathbb{N}$, which satisfy,
    \begin{equation} \label{eq:HLB}
        \frac{R}{P} \leq
        \mathcal{E}(r^{(A)} + n^2/P, r^{(B)} + n^2/P, c^{(C)} + n^2/P).
    \end{equation}
    Lemma~\ref{lem:MMExp}
    provides an expansion bound $\mathcal E$ for Strassen's algorithm, yielding
    \begin{align*}
        \frac{R}{P} \leq
        \min\big( c^{(A)} + n^2/P, c^{(B)} + n^2/P, c^{(C)} + n^2/P \big)^{\log_2(3)}.
    \end{align*}
    Since the rank is $R=n^{\log_2(7)}$, the communication cost is at least
    \begin{align*}
        c^{(A)} + c^{(B)} + c^{(C)}
        \geq
        3 \cdot \Big( \frac{n^{\log_3(7)}}{p^{\log_3(2)}} - \frac{n^2}{P} \Big).
    \end{align*}
\end{proof}


\subsection{Convolution} \label{sec:ConvSec}
Given a set of distinct nodes $\{x_i\}_{i=1}^m$, where $x_i \in \mathbb{C}$, the corresponding
Vandermonde matrix is $\B{V}_m^n \in \mathbb{C}^{m \times n}$, where
$[\B{V}_m^n]_{i,j} = x_i^{j-1}$. To compute the discrete convolution between two vectors $\B{f}, \B{g} \in \mathbb{C}^k$,
we use the Toom-$k$ bilinear algorithm,
$\mathcal{F} = \Big( (\B{V}_{2k-1}^{k})^T, (\B{V}_{2k-1}^{k})^T, (\B{V}_{2k-1}^{2k-1})^{-1} \Big)$~\cite{ju2020derivation}. The term Toom-$k$ refers to a particular bilinear algorithm for computing the convolution between vectors of size $k$, and it belongs to a broader class of convolution algorithms known as Toom-Cook. For example, the discrete Fourier transform (DFT) is a special case of Toom-Cook, hence our
lower bounds apply to fast Fourier transform (FFT)-based approaches for convolution.

It should be noted convolution can be applied to both integer and polynomial multiplication (i.e., Toom-Cook). While they are similar, a key difference between the two is that the former includes a carry-over step. While our analysis does not consider this carry-over step,  we mention that this step does not impact the communication lower bound of integer multiplication as in previous analysis~\cite{de2019complexity}.

By nesting Toom-Cook bilinear algorithms, one can utilize split-nesting schemes to derive algorithms for 1D convolution that
are more stable and computationally efficient~\cite{selesnick1994extending,agarwal1977new}, as well as to compute
multidimensional convolution~\cite{pitas1987multidimensional,lavin2016fast}. This, however, results in bilinear algorithms whose matrices are no longer
Vandermonde. In view of this, we apply Proposition~\ref{thm:nest_lb} to deduce rank expansion lower bounds for $\bigotimes_{i=1}^\tau (\B{V}_{2k_i-1}^{k_i})^T$ given arbitrary integers $k_i$. \editstwo{Since the $k_i$'s may be different, we define for notational convenience,
\begin{align} \label{eq:qgp_1}
  \underline{k} := \min_i k_i.
\end{align}}

\begin{lemma} \label{lem:ConvExp}
    \sloppy For
    a nested bilinear algorithm $ (\bigotimes_{i=1}^\tau \B{A}_i, \bigotimes_{i=1}^\tau \B{B}_i, \bigotimes_{i=1}^\tau \B{C}_i)$, where $\tau \in \mathbb{Z}_+$, $(\B{A}_i, \B{B}_i, \B{C}_i)$ is
    the bilinear algorithm for Toom-$k_i$, 
    then the function,
    \begin{equation*}
        \mathcal{E}(r^{(A)}, r^{(B)}, r^{(C)}) = \min(r^{(A)}, r^{(B)} )^{\log_{\underline k}(2\underline k-1)}, 
    \end{equation*}
    is an expansion bound function for the bilinear algorithm.
\end{lemma}
\begin{proof}
    Since $\B{C}_i = \B{V}_{2k_i-1}^{2k_i-1}$ is invertible, the function 
    \[
        \gs_C(\ell) 
        = 
        \ell 
        \leq 
        \min_{\B{P} \in \mathcal{P}_{2k_i-1}^{(\ell)}} \{ \mathrm{rank}(\B{C}_i\B{P}) \}
    \]
    \edits{is a rank expansion lower bound by Definition~\ref{def:RankExpLBDef}}. 
    
    Furthermore, $\B A_i = \B B_i = \B{V}_{2k_i-1}^k$ are Vandermonde matrices with unique nodes $ \{x_j\}_{j=1}^{2k_i-1}$. Thus, $(\B{V}_{2k_i-1}^{k_i})^T\B P$, 
    is the transpose of a Vandermonde matrix with a subset of nodes from $\{x_i\}$, which must also consist of unique nodes. Hence, $(\B{V}_{2k_i-1}^{k_i})^T\B P$ is full rank. Using this observation and the fact $a^{\log_{2x-1}(x)}$ is increasing w.r.t. $x \geq 1$ for any fixed $a \geq 1$ \editstwo{as well as recalling $\underline{k}$ from~\eqref{eq:qgp_1}}, one can directly verify
    \edits{\[  
        \sigma(\ell) 
        = 
        \ell^{\log_{2\underline{k}-1}(\underline{k})}
        \leq
        \min_{\B{P} \in \mathcal{P}_{2k_i-1}^{(\ell)}} \{ \mathrm{rank}(\B{A}_i\B{P}) \},
    \]
    is a rank expansion lower bound for $\B{A}_i$ as well as for $\B{B}_i$. }
    
    Like in the proof for Lemma~\ref{lem:MMExp}, one can verify $\gs$ is concave, log-log concave, and satisfies $\gs(0) = 0$, as does $\gs_C$. Also, $d_A \equiv \gs^\dagger(1) = 1$. Thus, one can recursively apply Theorem~\ref{thm:nest_lb} to confirm $\sigma$ is a rank expansion lower bound for $\bigotimes_{i=1}^\tau \B{A}_i$ and $\bigotimes_{i=1}^\tau \B{B}_i$, and $\sigma_C$ is for $\bigotimes_{i=1}^\tau \B{C}_i$. Thus, with $K = \prod_i k_i$, any $\ell \in [K]$, and all $\B{P} \in \mathcal{P}_{K}^{(\ell)}$,
    \begin{align*}
        \sigma(\ell) 
        &= \ell^{\log_{2 \underline k-1}(\underline k)} 
        \leq 
        \min\big\{
        \mathrm{rank}\big([\otimes_{i=1}^\tau \B{A}_i]\B P\big),
        \mathrm{rank}\big([\otimes_{i=1}^\tau\B{B}_i]\B P\big)
        \big\}.
    \end{align*}

    Applying the monotone function $\sigma^{-1}(\ell) = \ell^{\log_{\underline k}(2\underline k-1)}$ to both sides and recalling the choice of the expansion bound function $\mathcal{E}$, we get for all $\B{P} \in \mathcal{P}_{K}^{(\ell)}$,
    \begin{align*}
        \#\mathrm{cols}(\B P) 
        &\leq 
        \mathcal{E}\Big(
        \mathrm{rank}\big([\otimes_{i=1}^\tau \B{A}_i]\B P\big),
        \mathrm{rank}\big([\otimes_{i=1}^\tau\B{B}_i]\B P\big),
        \mathrm{rank}\big([\otimes_{i=1}^\tau\B{C}_i]\B P\big)
        \Big).
    \end{align*}
    Thus, $\mathcal{E}$ is an expansion bound by Definition~\ref{def:ExpansionBnd} (c.f.~\eqref{eq:rnk_exp_to_exp_bnd}).
\end{proof}
We now state a communication lower bound when nesting the same Toom-$k$ bilinear algorithm. Note that these bounds hold for both multidimensional and recursive (1D) convolution. The latter is simply multidimensional convolution plus a recomposition step, where the recomposition is applied via a pre-multiplication by the linear operator $\B Q \in \{0,1\}^{2n-1 \times (2k-1)^\tau}$ (see Section 7.3 in~\cite{ju2020derivation}) to $\B{C}^{\otimes \tau}$. Because the rank of the $\B C^{\otimes \tau}$ matrix, the only matrix in the bilinear algorithm that is affected by the recomposition matrix $\B Q$, is absent in the expansion bound of Lemma~\ref{lem:ConvExp}, the expansion bound still holds for recursive (1D) convolution.

First, we focus on sequential communication, or communication between a fast and slow memory (e.g., cache and memory). In the special 1D case, where Toom-$n$ is recursively applied to an input vector of size $N=n^d$, the following result matches previously established lower bounds~\cite{bilardi2019complexity}. The proof is a direct result from Proposition~\ref{prop:edgar_seq_thm} and Lemma~\ref{lem:ConvExp} (and follows similarly to the proof for Corollary~\ref{cor:MMVLB}). Hence, we skip the proof.

\begin{corollary}
\label{cor:ConvVLB}
   Given two $d$-dimensional tensors where each mode length is $n$ and \edits{fast memory of size $M$,} the sequential communication cost of discrete convolution using a nested Toom-n bilinear algorithm is at least
    \begin{equation*} 
      \max \bigg\{
      \frac{(2n-1)^d}{M^{\log_n(2n-1)}} \cdot M, 2n^d + (2n - 1)^d \bigg \}.
  \end{equation*}
\end{corollary}


Next, we derive parallel communication (i.e., communication between parallel processors) lower bounds for multidimensional convolution, which asymptotically match previously established lower bounds~\cite{de2019complexity}. Again, the result follows from Proposition~\ref{prop:edgar_par_thm} and Lemma~\ref{lem:ConvExp}, hence we skip the proof.

\begin{corollary}
\label{cor:ConvHLB}
  Given two $d$-dimensional tensors where each mode length is $n$ and \edits{$P$ processors,} the parallel communication costs of discrete convolution using a nested Toom-n bilinear algorithm is at least
  \begin{equation*}
      2 \cdot \Big( \frac{n^d}{P^{\log_{2n-1}(n)}} - \frac{n^d}{P} \Big).
  \end{equation*}
\end{corollary}

\subsection{Partially Symmetric Tensor Contractions} \label{sec:SymTSec}

\editstwo{A tensor contraction, which generalizes matrix multiplication, are tensor products summed over a subset of modes (indices) from the tensors. Tensor contractions are also equivalent to a matrix multiplication between the unfolded tensors, where a subset of tensor modes are mapped to either row or column indices of a matrix~\cite{golub2013matrix,solomonik2017communication}.}
However, special fast bilinear algorithms become possible when the tensors have symmetry~\cite{solomonik2015contracting}.
Symmetric tensors are \editstwo{invariant} under all permutations of their indices, e.g., $t_{ijk}=t_{jik}=\ldots$.
For example, a product of a symmetric matrix $\B A$ and a vector $\B b$ can be computed from $n(n+1)/2$ scalar products, mostly of the form $a_{ij}(b_i+b_j)$, yielding a bilinear algorithm with rank $R = n(n+1)/2$.
The main application of such symmetry preserving algorithms is to lower the cost of contraction algorithms for partially symmetric tensors (tensors that are equivalent only under permutations of a subset of their indices).
For example, given a partially symmetric tensor $\B{\mathcal{T}}$ with symmetry $t_{ijab}=t_{jiab}$, the contraction $u_{iac} = \sum_{jb} t_{ijab} v_{jbc}$ can be performed with \edits{2 times} fewer operations to leading order than in the nonsymmetric case, by nesting a symmetry preserving algorithm for a symmetric vector product (corresponding to contraction indices $i$ and $j$) with an algorithm for matrix multiplication (corresponding to contraction indices $a,b,c$).

Tensors with partial symmetry are prevalent in quantum chemistry methods, a core application domain of higher-order tensor contractions~\cite{hirata2003tensor}.
By analysis of these bilinear algorithms' rank expansions, communication lower bounds have been established showing that such symmetry preserving algorithms  require asymptotically more communication for some contractions  of symmetric tensors~\cite{solomonik2017communication}.
Our results allow us to derive the first communication lower bounds for nested symmetry preserving algorithms, yielding lower bounds on communication costs for symmetry preserving contraction algorithms on partially symmetric tensors.

The lemma below follows from the analysis in the proof of Lemma 6.3 in~\cite{solomonik2017communication}.
\begin{lemma}
\label{lem:sym_rank_bounds}
For the bilinear algorithm $(\B A, \B B, \B C)$ corresponding to symmetry preserving contraction of symmetric tensors of order $s+v$ and $v+t$ over $v$ indices, we can lower bound the rank expansion of each encoding matrix as follows:
$\sigma_A(k)=k^{(s+v)/(s+t+v)}/{ s+t+{v\choose t}}$,
$\sigma_B(k)=k^{(v+t)/(s+t+v)}/{ s+t+{v\choose s}}$, and
$\sigma_C(k)=k^{(s+t)/(s+t+v)}/{ s+t+{v\choose v}}$.
\end{lemma}

We derive lower bounds for the nesting of multiple symmetry preserving algorithms, as well as nesting of a symmetry preserving algorithm with a nonsymmetric contraction algorithm.
In the former case, we consider nesting of two arbitrary symmetry preserving algorithms.
\begin{lemma}
\label{lem:sym_nest_exp}
For the bilinear algorithm $(\B A \otimes \B U , \B B \otimes \B V, \B C \otimes \B W)$, where $(\B A, \B B, \B C)$ is a symmetry preserving contraction of symmetric tensors of order $s+v$ and $v+t$ over $v$ indices and all dimensions equal to $n_1$, while $(\B U, \B V, \B W)$ is a symmetry preserving contraction of symmetric tensors of order $s'+v'$ and $v'+t'$ over $v'$ indices with all dimension equal to $n'_2$,
we can lower bound the rank expansion of $\B A\otimes \B U$
by
\begin{align*}
\sigma_{A\otimes U}(k)&\geq\min_{\substack{
k_1\in[1,n_1^{s+t+v}],k_2\in[1,n_2^{s'+t'+v'}], \\  k_1k_2 \geq k
}}\frac{1}{{ s+t+v\choose t}{ s'+t'+v'\choose t'  }}
\cdot k_1^{\frac{s+v}{s+t+v}} \cdot k_2^{\frac{s'+v'}{s'+t'+v'}} \\
& \geq \frac{1}{{ s+t+v\choose t}{ s'+t'+v'\choose t'  }}\cdot k^{\min\big(\frac{s+v}{s+t+v},~\frac{s'+v'}{s'+t'+v'}\big)},
\end{align*}
as well as similar bounds for $\B B \otimes \B V$ and $\bC \otimes \B W$.
\end{lemma}
\begin{proof}
The theorem follows by application of Theorem~\ref{thm:nest_lb} 
on the rank expansion lower bounds given by Lemma~\ref{lem:sym_rank_bounds}. Note that in the first inequality, we restrict $k_A$ and $k_B$ to $[1,n_1^{s+t+v}]$ and $[1,n_2^{s'+t'+v'}]$, respectively, which weakens the bound up to a constant.
\end{proof}
The rank expansion lower bound $\sigma_{A\otimes U}$ in Lemma~\ref{lem:sym_nest_exp} generalizes to nestings of three or more symmetry preserving bilinear algorithms.
This rank expansion lower bound implies parallel and sequential communication for nested bilinear algorithms follow immediately from those of the nested parts.
These communication lower bounds ascertain that standard approaches for tiling nested loops can asymptotically minimize communication done in the execution of the bilinear algorithm.


We also consider nestings of a symmetry preserving algorithm with a standard (nonsymmetric) tensor contraction.
Nonsymmetric tensor contractions are equivalent to matrix-matrix products with an appropriate choice of matrix dimensions.
Like for symmetric tensor contractions, we assume tensor dimensions are equal size and classify contractions by the tuple $(s,t,v)$.
When one of $s,t,v$ is zero, the number of products needed to compute the contraction matches the size of the largest tensor ($n^{s+t+v}=n^{\max(s+t,s+v,v+t)}$), and the corresponding bilinear encoding matrix is (some permutation of) the identity matrix, with a rank expansion of $\tilde{\sigma}(k)=k$.
When this is the case for nonsymmetric contractions, tight communication lower bounds thereof can be derived just by considering the rank expansion of a single matrix.
For the symmetry preserving algorithm, for any choice of $s,t,v$, a tight communication lower bound can be derived from the rank expansion of just one of the matrices.
For further details on these lower bounds and optimal algorithms, see~\cite{solomonik2017communication} (Section 7 discusses upper bounds).
Consequently, we can use our general lower bounds on the rank expansion of a Kronecker product of matrices to derive communication lower bounds that we expect are tight by restricting the type of nonsymmetric contraction performed.
\begin{lemma}
\label{lem:sym_ns_nest_exp}
\editstwo{Consider} the bilinear algorithm $(\B A \otimes \B U , \B B \otimes \B V, \B C \otimes \B W)$, where $(\B A, \B B, \B C)$ is a symmetry preserving contraction of symmetric tensors of order $s+v$ and $v+t$ over $v$ indices and all dimensions equal to $n_1$, while $(\B U, \B V, \B W)$ is a contraction of nonsymmetric tensors of order $s'+v'$ and $v'+t'$ over $v'$ indices with all dimension equal to $n'_2$.
If $t'=0$, we can lower bound the rank expansion of $\B A\otimes \B U$
by
\begin{align*}
\sigma_{A\otimes U}(k)&=\sigma_{A}(k)=\frac{1}{{ s+t+v\choose t}}
\cdot k^{(s+v)/(s+t+v)}.
\end{align*}
\editstwo{Similar bounds hold} for $\B B \otimes \B V$ (if instead of $t'=0$, we have $s'=0$) and $\bC \otimes \B W$ (if instead of $t'=0$, we have $v'=0$).
\end{lemma}

The expansion bound in Lemma~\ref{lem:sym_ns_nest_exp} implies that we can obtain a lower bound on parallel and sequential communication cost on nested algorithms composed of symmetry preserving algorithms and a nonsymmetric tensor contraction where one of $s',t',v'$ is zero.
Given a rank expansion lower bound for one of the inputs or the output, e.g., $\sigma_{A\otimes U}$, we obtain a sequential communication lower bound of the form,
\[\Omega(n^{s+t+v+s'+t'+v'}M/\sigma_{A\otimes U}^{-1}(M)),\]
and a parallel communication lower bound of
\[\Omega(\sigma_{A\otimes U}(n^{s+t+v+s'+t'+v'}/p) - n^{s+v+s'+v'}/p). \]
For nested symmetry preserving algorithms, the greatest of the three communication lower bounds (based on $\sigma_{A\otimes U}$, $\sigma_{B\otimes V}$, or $\sigma_{C\otimes W}$) would be asymptotically attainable for sufficiently small $M,p$ with standard approaches for multidimensional loop tiling~\cite{christ2013communication}.
While the bound Lemma~\ref{lem:sym_ns_nest_exp} should also be asymptotically attainable for many contractions, but not all, as preclusion of $t'>0$ implies we do not provide bounds for communication associated of all inputs/outputs in a particular contraction.
The new communication lower bounds imply that for some partially symmetric tensor contractions, the use of the symmetry preserving algorithm may require asymptotically more communication than if the symmetry was ignored.
However, these contractions involve high order tensors.
The example below is among the simplest possible cases.
\begin{example} \label{ex:example_2}
Consider the contraction,
\[c_{im} = \sum_{j,k,l} a_{ijkl} b_{jklm},\]
where $a_{ijkl}$ is symmetric under any permutation of $(i,j,k,l)$ and $b_{jklm}$ is symmetric under any permutation of $(j,k,l)$.
Here, a symmetry preserving algorithm with $s=1, v=3, t=0$ may be nested with a nonsymmetric contraction with $s'=0, v'=0, t'=1$.
By Lemma~\ref{lem:sym_ns_nest_exp}, we obtain the following rank expansion lower bound,
\[\sigma_{B\otimes V}(k)=\frac{1}{{s+t+v \choose s}}k^{(t+v)/(s+t+v)}= (1/4)k^{3/4}.\]
Assume the dimension of each mode of the tensor (range of each index) is $n_1=n_2=n$.
Overall, this algorithm would then require \edits{4 times} fewer products ($n^5/24$ to leading order) than if only considering symmetry in $(j,k,l)$ and performing classical matrix multiplication with dimensions $n \times {n+2 \choose 3}\times n$ (requiring $n^5/6$ products to leading order).
However, by our new lower bounds, it would require $\Omega(n^5/M^{1/3})$ sequential communication.
On the other hand, for sufficiently small $M$,
the matrix-multiplication-based approach requires only $O(n^5/M^{1/2})$ sequential communication.
\end{example}

\section{Conclusion} \label{sec:conclusion}
We develop a new framework to ascertain communication lower bounds for any bilinear algorithm via the rank expansion of a matrix, or the minimum rank of any submatrix of fixed size. Unlike previous works which assume a particular computational DAG, our lower bounds consider a larger space of permissible computational DAGs. Our new communication lower bounds for recursive convolution match previous bounds~\cite{bilardi2019complexity,de2019complexity}, suggesting that any algebraic reorganization of convolution cannot reduce communication costs in this setting.

We note two limitations in our analysis that prevents us from obtaining tight lower bounds for standard nested matrix multiplication as well as some partially symmetric tensor contractions. \editstwo{First, as described in Example~\ref{ex:eqgp_1}, we separately bound the ranks of the matrices from the nested bilinear algorithm to derive expansion bounds. However, this can derive expansion bounds which are much larger than the expansion function (which bounds the rank simultaneously), especially when the matrices have an extremely low-rank structure. Second, there seems to be a gap between converting rank expansion bounds to communication complexity. As discussed in Remark~\ref{rem:rqgp_1}, the proof of the communication lower bound via rank expansions assumes if there is a low-rank structure in a matrix of a bilinear algorithm, that low-rank structure persists in any subset of columns, which is not true for Strassen's algorithm. Therefore, to get a tight lower bound, we hypothesize one needs to separately handle low-rank and nearly full-rank subsets of columns in the matrices.}

\section*{Acknowledgements}
The authors would like to thank the anonymous reviewers for the comments that significantly improved the presentation and clarity of this work.

CJ is supported by the U.S. Department of Energy, Office of Science, Office of Advanced Scientific Computing Research, Department of Energy Computational Science Graduate Fellowship under Award Number DE-SC0022158. 

\textbf{Disclaimer}. This report was prepared as an account of work sponsored by an agency of the United States Government. Neither the United States Government nor any agency thereof, nor any of their employees, makes any warranty, express or implied,or assumes any legal liability or responsibility for the accuracy, completeness, or usefulness of any information, apparatus, product, or process disclosed, or represents that its use would not infringe privately owned rights. Reference herein to any specific commercial product, process, or service by trade name, trademark, manufacturer, or otherwise does not necessarily constitute or imply its endorsement, recommendation, or favoring by the United States Government or any agency thereof. The views and opinions of authors expressed herein do not necessarily state or reflect those of the United States Government or any agency thereof.

\formatmode{
\bibliographystyle{unsrt}
\bibliography{refs}
}{
\bibliographystyle{spmpsci}      
\bibliography{refs}   
}

\appendix
\section{Appendix: Improved Lower Bounds by Limiting the Grid 
Expansion}
\label{appendix:improved_bnd}

\normalsize

Let $\bA \in \C^{m_A \times n_A}$, $\bB \in \C^{m_B \times n_B}$, and $\bC = \bA \otimes \bB$.
As mentioned after the main theorems in Section \ref{sec:main_results}, bounds $\gs_C$ in those results require defining $\gs_A$ and $\gs_B$ beyond $n_A$ and $n_B$.
This extrapolation can lead to a loose lower bound $\gs_C$ even if $\gs_A$ and $\gs_B$ are tight.
We illustrate this phenomenon in the next example.

\begin{example} \label{exp:avoid_extension}
    Consider the case $\bA = \bB \in \R^{4 \times 7}$ with rank and Kruskal rank (maximum $k$ such that any $k$ different columns are linearly independent~\cite{kruskal1977three}) being 4.
    For example,
    \begin{equation*}
    \bA = \bB =
        \begin{bmatrix*}[r]
            1&0&0&0&1&1&1\\
            0&1&0&0&1&2&3\\
            0&0&1&0&1&4&9\\
            0&0&0&1&1&8&27
        \end{bmatrix*}.
    \end{equation*}
    Denote the $i$th columns of $\bA$ and $\bB$ as $\ba_i$ and $\bb_j$, respectively. Now $\tilde{\gs}_A(x) = \tilde{\gs}_B(x) = \min\set{x, 4}$. Let $\bC = \bA \otimes \bB$, and we seek a rank expansion lower bound $\gs_C(k)$ for $\bC$.
    
    When $k = 13$, it is not hard to check that $\tilde{\gs}_C(k) = 7$, which is attained by submatrix $\bC \B P = \set{\ba_i \otimes \bb_j : i = 1 \text{ or } j = 1}$. 
    If we naturally take $\gs_A(x) = \gs_B(x) = \min\set{x, 4}$ (on $\R$), Theorem~\ref{thm:lb} gives $\gs_C(13) = 4$. 
    If we take instead $\gs_A(x) = \gs_B(x) = x^{\ln 4 / \ln 7}$, Theorem~\ref{thm:lb} gives $\gs_C(13) \approx 6.2$, which is optimal after rounding up to integer.
    Although $x^{\ln 4/\ln 7}$ is not as tight as $\min\set{x, 4}$ in the range 
    $x \in [0,7]$,
    its 
    extrapolation on $x \geq 7$ is greater than that of $\min\set{x, 4}$.
    
    Indeed, when $x \leq 7$, using $\min\set{x, 4}$ offers a tighter bound than $x^{\ln 4/ \ln 7}$.
    For example, when $k = 5$, with $\gs_A(x) = \gs_B(x) = \min\set{x, 4}$, $\gs_C(5) = 4 = \tilde{\gs}_C(5)$, but with $\gs_A(x) = \gs_B(x) = x^{\ln 4 / \ln 7}$, $\gs_C(5) \approx 3.1$. 
\end{example}

We see from above that a tighter rank expansion lower bounds on $\bA$ and $\bB$ may lead to a looser rank expansion lower bound $\gs_C$.
Thus, there is an opportunity to improve the established bound in Theorem \ref{thm:lb}.
To achieve this improvement and avoid the behavior in the example above, we derive a lower bound $\gs_C$ that evaluates $\gs_A$ and $\gs_B$ on their intended domains of $[0,n_A]$ and $[0,n_B]$, respectively.

\subsection{The $L$-shaped bound}\label{sec:L_bnd}

The derivation of the new rank expansion lower bound $\gs_C$ is not much different from the main theorems.
The only difference is in the continuous relaxation step (Section \ref{sec:cont_step}).
Recall that in the discrete step (see Section~\ref{sec:discrete_analysis} and definitions therein), we constructed the pre-CDG $S$ that is a subgrid of the basis of $D = \vcl(G)$.
Thus, we know 
\begin{equation*}
    \cdg(S) \sset [0, \gs_A(n_A)] \times [0, \gs_B(n_B)].
\end{equation*}
To simplify the notation, we will denote 
\[
    r_A := \gs_A(n_A),\ \ r_B := \gs_B(n_B),\ \ S := \cdg(S).
\]
In the continuous relaxation step, we carried out $\merge(2, 3)$ (Definition~\ref{def:merge_step}) until $S$ becomes a 2-step CDG, and then $\merge(1, 2)$ to make it a rectangle.
In fact, before the last step $\merge(1, 2)$, the entire grid remains a subgrid of $[0, r_A] \times [0, r_B]$.
To avoid leaving this domain (so that $\spcexp(S)$ remains in $[0, n_A] \times [0, n_B]$), 
we stop $\merge(1,2)$ if either the first or second step hits the boundary of the domain, which leads to an $L$-shaped grid (a 2-step CDG). 
We refer to this stoppage of $\merge(1,2)$ as \textit{early stopping}.

\begin{defn} \label{def:L_shaped}
    The $L$-\textit{shaped} CDG $S = L(x_1, y_1; x_2, y_2)$ is the 2-step CDG with horizontal edges at $y_1,~ y_2$ and vertical edges at $x_1,~ x_2$. By convention, we require $0 < x_1 < x_2$, $0 < y_1 < y_2$.
\end{defn}

\begin{defn} \label{def:LR}
    Fix the values $r_A = \gs_A(n_A)$ and $r_B = \gs_B(n_B)$. We denote the collection of $L$-shaped grids of size $t$ that touch the boundaries as
    \begin{equation*}
        \cL(t) = \set{L = L(x_1, y_1; r_A, r_B): \vert L \vert = t}.
    \end{equation*}
    Denote the collection of rectangle grids of size $t$ within $[0, r_A]\times[0, r_B]$ as
    \begin{equation*}
        \cR(t) = \{R = [0, x] \times [0, y]: x \in [1, r_A],\  y \in [1, r_B],\  \vert R \vert = t\}.
    \end{equation*}
\end{defn}

The $L$-shaped bound we derive in this section extends the bounds in the main theorems. 
Thus, we will work under the same assumptions on $\gs_A$ and $\gs_B$. See equations \eqref{eq:assumpt_zero} and \eqref{eq:assumpt_smooth}.

It is not hard to check that the early-stopped merge also results in an upper bound of $\ab{S}$ as stated in the lemma below, which serves as the analogue of Lemma \ref{lem:continuous_step}.

\begin{lemma} \label{lem:newmerge}
    Let $S$ be a CDG in $[0, r_A] \times [0, r_B]$ of size $|S| = t$. Then
    \begin{equation} \label{eq:new_continous_step}
        |\spcexp(S)| \leq \max_{M \in \cL(t) \cup \cR(t)} \ab{M}.
    \end{equation}
\end{lemma}

\color{black}

\begin{proof}
    The proof is similar to that of Lemma~\ref{lem:continuous_step}, 
    where as long as the grid has at least three steps, we apply $\merge(2, 3)$ repeatedly.
    This produces an $L$-shaped grid within $[0, r_A] \times [0, r_B]$.
    However, in the last merge operation, $\merge(1, 2)$, we stop increasing $u$, the height difference between steps 1 and 2, if step 1 reaches height $r_B$ before step 2 reaches the ground.
    This may give an $L$-shaped grid $L(x_1, y_1; x_2, r_B)$. 
    By Lemma \ref{lem:continuous_step}, either this $L$-shaped grid or a rectangle grid is an upper bound of $\ab{S}$.
    If the rectangle one is an upper bound, then we are done since it is in $\cR(t)$.

    Next, consider the case the upper bound is $L(x_1, y_1; x_2, r_B)$.
    Similar to how we increased/decreased the height of the steps, 
    we can horizontally merge the two horizontal layers of the 2-step stair, and we similarly disallow the width of the lower layer to go beyond $r_A$. 
    This may result in an $L$-shaped grid $L(x_1', y_1, r_A, r_B)$. 
    After these two merges, the resulting grid may either be a rectangle from $\cR(t)$, or the $L$-shaped grid above, which is from $\cL$. The proof is thus complete.
\end{proof}

With this new continuous relaxation step, we analogously derive the corresponding bound $\gs_C$, which includes the $L$-shaped bound, in the following theorem.

\begin{theorem} \label{thm:newlb}
    Suppose functions $\gs_A$ and $\gs_B$ are concave rank expansion lower bounds of $\bA \in \mathbb{C}^{m_A \times n_A}$ and $\bB\in \mathbb{C}^{m_B \times n_B}$, respectively, with $\gs_A(0) = \gs_B(0) = 0$.
    Let $r_A = \gs_A(n_A)$, $r_B = \gs_B(n_B)$, $d_A = \gs_A^{\dagger}(1)$, and $d_B = \gs_B^{\dagger}(1)$. Define
    \begin{equation*}
        R_C(k) =
        \min_{\substack{k_A\in [d_A, n_A],~k_B \in [d_B, n_B],\\ k_A k_B \geq k}}
          \sigma_A(k_A) \cdot \sigma_B(k_B),
    \end{equation*}
    and
    \begin{equation} \label{eq:l_bnd}
        L_C(k) = \min_{\substack{k_A \in [0, n_A],~k_B \in [0, n_B],\\
        k_{A} r_B + k_B r_A - k_A k_B = k}}
        \gs_A(k_{A})r_B + \gs_B({k_{B}}) r_A - \sigma_A(k_{A}) \cdot \sigma_B(k_{B}).
    \end{equation}
Then $\gs_C(k) = \min\set{L_C(k), R_C(k)}$ is a rank expansion lower bound of $\bC = \bA \otimes \bB$.
When $R_C(k) \leq \max\set{r_A, r_B}$, $R_C$ is a rank expansion lower bound for $\bC$.
\end{theorem}

\begin{proof}
Again, using a density argument on the functions, we hereafter assume $\gs_A$ and $\gs_B$ are strictly increasing smooth functions so that they are invertible.
We start by showing $\gs_C$ is continuous and increasing, as required by the definition of a rank expansion lower bound.
Since $R_C$ and $L_C$ are continuous, so is $\sigma_C$.  Clearly, $R_C$ is increasing. 
As for $L_C$, we have the equivalent definition,
\begin{equation*}
    L_C(k) = \min_{\substack{k_A \in [0, n_A],~ k_B \in [0, n_B],\\ (n_A - k_A)(n_B - k_B)\, = \,n_A n_B - k}} 
    r_A r_B - \left(r_A - \gs_A(k_A)\right)\left(r_B - \gs_B(k_B)\right).
\end{equation*}
Thus, as $k$ increases, one can increase $k_A$ or $k_B$ so that $L_C$ increases. 
Hence, $L_C$ is also an increasing function. Consequently, $\gs_C$ is an increasing function.

To prove $\sigma_C$ is a lower bound of the rank, we express the right-hand side of \eqref{eq:new_continous_step} as a function $\phi(t)$ as in the proof of Theorem~\ref{thm:lb}.
We compute the maximal expansion size of stairs in $\cR(t)$ and $\cL(t)$ below: 
\begin{equation} \label{eq:phi_RL}
    \begin{split}
        \phi_R(t) &= \max_{\substack{t_A \in [1, r_A],~ t_B \in [1, r_B],\\ t_A t_B = t}} \gs_A^{-1}(t_A)\gs_B^{-1}(t_B);\\
        \phi_L(t) &= \max_{\substack{x_1 \in [0, r_A],~ y_1 \in [0, r_B],\\ r_Bx_1 + r_Ay_1 - x_1y_1 = t}} n_A\gs_B^{-1}(y_1) + n_B\gs_A^{-1}(x_1) - \gs_A^{-1}(x_1)\gs_B^{-1}(y_1).
    \end{split}
\end{equation}
Then we have
\begin{equation*}
    \phi(t) := \max\set{\phi_R(t), \phi_L(t)} = \max_{M \in \cL(t) \cup \cR(t)} \ab{M}.
\end{equation*}
When $k \leq d_Ad_B$, $\gs_C(k) \leq R_C(k) = 1$, which is clearly a lower bound of the rank of a nonzero matrix.
It remains to show the proposed function $\gs_C$ satisfies $\gs_C = \phi^{-1}$ on $[d_Ad_B, +\infty)$.

\textit{Step 1. $R_C = \phi_R^{-1}$ on $[d_Ad_B, +\infty)$.} 
This proof is the same as the proof for Lemma \ref{lem:invphi_main}. 
All arguments carry through with the new constraints $t_A \leq r_A$, $t_B \leq r_B$, and $k_A \leq n_A$, $k_B \leq n_B$.

\textit{Step 2. $L_C = \phi_L^{-1}$.}
We will repeat the proof for Lemma \ref{lem:invphi_main} for this $L$-shaped case.
First, we show $L_C \leq \phi_L^{-1}$.
Assume by contradiction $\phi_L(\gs_C(k)) > k$ for some $k$. Then there exist $x_1 \in [0, r_A],~ y_1 \in [0, r_B]$, and $r_B x_1 + r_A y_1 - x_1 y_1 = \gs_C(k)$ such that
\begin{equation*}
    k' = n_A\gs_B^{-1}(y_1) + n_B\gs_A^{-1}(x_1) - \gs_A^{-1}(x_1)\gs_B^{-1}(y_1) > k.
\end{equation*}
Consequently, as $L_C$ is strictly increasing,
\begin{align*}
    \gs_C(k) 
    &\leq 
    L_C(k) < L_C(k')\\
    &=
    \min_{\substack{k_A\in [0, n_A],~k_B \in [0, n_B],\\ n_B k_A + n_A k_B - k_A k_B = k'}}
    r_B \gs_A(k_A) + r_A \gs_B(k_B) - \sigma_A(k_A) \cdot \sigma_B(k_B) \\
    &\leq
    r_B \gs_A\left(\gs_A^{-1}(x_1)\right) + r_A \gs_B\left(\gs_B^{-1}(y_1)\right) - \sigma_A\left(\gs_A^{-1}(x_1)\right) \cdot \sigma_B\left(\gs_B^{-1}(y_1)\right) \\
    &=
    r_B x_1 + r_A y_1 - x_1 y_1 \\
    &= 
    \gs_C(k).
\end{align*}
This is absurd. 
Next, we show the other direction, $\phi_L^{-1} \leq L_C$, by showing that if $L_C(k) = t$, then $\phi_L(t) \geq k$.
Indeed, if $L_C(k) = t$, then let $x_1 = \gs_A(k_A) \in [0, r_A]$ and $y_1 = \gs_B(k_B) \in [0, r_B]$. Then we have both
\begin{align*}
    \gs_A^{-1}(x_1) r_B + \gs_A^{-1}(y_1) r_A - \gs^{-1}_A(x_1) \gs_B^{-1}(y_1) &= k, \\ 
    r_Bx_1 + r_Ay_1 - x_1y_1 &= t.
\end{align*}
Therefore,
\begin{equation*}
    \phi_L(t) 
    \geq
    \gs_A^{-1}(x_1) r_B + \gs_A^{-1}(y_1) r_A - \gs^{-1}_A(x_1) \gs_B^{-1}(y_1) 
    =
    k. 
\end{equation*}
Hence, $L_C = \phi_L^{-1}$.

We now conclude $\gs_C = \min\set{R_C, L_C}$ is the inverse of $\phi = \max\set{\phi_R, \phi_L}$.
This is straightforward with steps 1 and 2, since all 4 functions here are positive increasing functions on $\R_+$.
We have established that $\gs_C$ is a valid rank expansion lower bound for $\bC$.

Finally, to see why when $R_C(k) \leq \max\set{r_A, r_B}$ implies that $R_C$ is a rank expansion lower bound, we consider the final merge operation $\merge(1, 2)$.
The $L$-shaped bound will only come into play if we need to increase the height difference $u$, and step 1 reaches height $r_B$ before step 2 goes down to the ground.
However, each of the 2 steps is of width at least 1, since the original stair $S$ corresponds to a CDG.
This means step 1 will never reach $r_B$ before step 2 reaches the ground if $|S| \leq r_B$.
Similarly, since $r_A$ and $r_B$ are interchangeable (by considering $\bC = \bB \otimes \bA$), we see that early stopping the merge and the $L$-shaped grid will never come into play if $|S| \leq \max\set{r_A, r_B}$.

Suppose $G$ is a grid and it produces a (pre-) CDG $S$ with size $|S| \leq \max\set{r_A, r_B}$.
Then since we do not need to early stop the merge,
\begin{equation*}
    \phi_R(\rank(G)) \geq \phi_R(|S|) \geq |\spcexp(S)| \geq |G|. 
\end{equation*}
Hence, $\rank(G) \geq \phi_R^{-1}(k) = R_C(k)$, which establishes that $R_C$ is a rank expansion lower bound in this case.
\end{proof}

The new bound in Theorem~\ref{thm:newlb} is more complicated than Theorem \ref{thm:lb}, which does not involve $L$-shaped grids.
However, the new bound derived in this section is often tighter.
Indeed, let $\gs_C^{R}$ and $\gs_C^{R+L}$ be, respectively, the rank expansion lower bounds derived from Theorem \ref{thm:lb} and Theorem \ref{thm:newlb}.
Recall the $\phi$ function we used when proving Theorem \ref{thm:lb}, shown below
\begin{equation}
    \phi^{\text{prev}}(t) = \max_{\substack{t_A,~ t_B \geq 1,\\ t_A t_B = t}} \gs_A^{-1}(t_A)\gs_B^{-1}(t_B).
\end{equation} 
Denote $\phi^{\text{new}} = \max\set{\phi_R, \phi_L}$ as the new $\phi$ function used in the proof of Theorem \ref{thm:newlb} above.
Since $\gs_C^{R} = \left(\phi^{\text{prev}}\right)^{-1}$ and $\gs_C^{R+L} = \left(\phi^{\text{new}}\right)^{-1}$,
if $\phi^{\text{prev}} \geq \phi^{\text{new}}$, the new bound $\gs_C^{L+R}$ is then tighter (greater) than $\gs_C^{R}$.

When $\phi_L \leq \phi_R$, it is clear $\phi^{\text{prev}} \geq \phi^{\text{new}}$, since $\phi_R$ and $\phi^{\text{prev}}$ are maximizing the same function and $\phi_R$ has a smaller feasible region, so $\phi^{\text{prev}} \geq \phi_R = \phi^{\text{new}}$.
When $\phi_L > \phi_R$, this is not clearly true.
When the maximizer $x_1$ in $\phi_L$ is at least 1, then we can prove $\phi^{\text{prev}} \geq \phi_L$, so $\gs_C^{L+R} \geq \gs_C^{R}$.
We can require $x_1 \geq 1$ by early stopping the horizontal merge step in the proof of Lemma \ref{lem:newmerge}, but this will make the bound too complicated to state.

\color{black}

The bound in Theorem \ref{thm:newlb} does not require defining $\gs_A$, $\gs_B$ beyond $n_A$ and $n_B$.
Thus, it also resolves the undesirable phenomenon in Example \ref{exp:avoid_extension}.
It provides a tighter rank expansion lower bound for $\bC = \bA \otimes \bB$ when given tighter rank expansion lower bounds for $\bA$ and $\bB$. 
This is not always the case with Theorem~\ref{thm:lb} as shown in Example \ref{exp:avoid_extension}.

Suppose we have two rank expansion lower bounds $\gs_A,~\hat{\gs}_A$ for $\bA$ and $\gs_B,~\hat{\gs}_B$ for $\bB$. 
If $\gs_{A,B} \geq \hat{\gs}_{A, B}$ on $[0, n_{A, B}]$, 
then the corresponding functions $R_C \geq \hat{R}_C$ on $[0, n_An_B]$. 
If in addition $\gs_{A, B}(n_{A, B}) = \hat{\gs}_{A, B}(n_{A, B})$, 
for example, both are 
equal to the true rank of $\bA$ and $\bB$, then also $L_C \geq \hat{L}_C$ on $[0, n_An_B]$. 

Despite these properties of the new bound in Theorem \ref{thm:newlb},
Theorem~\ref{thm:lb} and Theorem~\ref{thm:nest_lb} give much simpler bounds and can be easily applied recursively to derive a lower bound on the rank expansion for $\bC = \bigotimes_{i=1}^p \bA_i,~p\geq 3$.

The log-log convexity assumption simplifies Theorem \ref{thm:lb}.
Similar assumptions can also simplify the bound in Theorem \ref{thm:newlb}.
Under such assumptions, we can show $R_C \leq L_C$, which removes the need of $L_C$ function.
The obtained rank expansion lower bound is thus tighter than the one in Theorem \ref{thm:lb} as discussed above.
This is the main topic of the next section.

\subsection{Simplifying the $L$-shaped bound} \label{sec:simplify_L_shape}

The main goal is to understand when $R_C \leq L_C$, and thus $\gs_C = R_C$.
We will show this is the case when $\gs_A$ and $\gs_B$ satisfy an appropriate log-log convexity condition (see Definition \ref{def:loglog} for the definition and Section \ref{sec:main_proof_2} for the basic properties).
We will continue using the notations introduced in the previous section.
Recall that $r_A = \gs_A(n_A)$ and $r_B = \gs_B(n_B)$.

\begin{theorem} \label{thm:reduce2R}
    Let everything be defined as in Theorem \ref{thm:newlb}. Suppose in addition functions 
    $r_A - \gs_A(n_A - x)$ and $r_B - \gs_B(n_B - x)$ are log-log convex on $(0, n_A)$ and $(0, n_B)$, respectively.
    Then, $R_C(k) \leq L_C(k)$, and consequently,
    \begin{equation} \label{eq:simpleNestBnd}
        \gs_C(k) = \min_{\substack{k_A\in [d_A, n_A],~k_B \in [d_B, n_B],\\ k_A k_B \geq k}}
          \sigma_A(k_A) \cdot \sigma_B(k_B)
    \end{equation}
    is a rank expansion lower bound of $\bC = \bA \otimes \bB$.
    This bound is tighter than the one in Theorem \ref{thm:lb}.
\end{theorem}

\begin{proof}
As defined in~\eqref{eq:phi_RL} during the proof of Theorem~\ref{thm:newlb}, the size of an $L$-shaped CDG after a grid expansion can be bounded by
\begin{equation*}
\phi_L(t) 
= 
\max_{\substack{x_1 \in [0, r_A],~ y_1 \in [0, r_B],\\ r_Bx_1 + r_Ay_1 - x_1y_1 = t}} n_A\gs_B^{-1}(y_1) + n_B\gs_A^{-1}(x_1) - \gs_A^{-1}(x_1)\gs_B^{-1}(y_1).
\end{equation*}
Using a density argument, we assume that $\gs_A^{-1}$ and $\gs_B^{-1}$ are smooth and strictly increasing on $[0, r_A]$ and $[0, r_B]$.
According to Theorem \ref{thm:newlb}, when $R_C(k) \leq \max\set{r_A, r_B}$, $R_C$ is indeed a rank expansion lower bound of $\bC$.
Thus, hereafter we only consider the case $R_C(k) = t \geq \max\set{r_A, r_B}$.
To show $R_C \leq L_C$, it suffices to show $\phi_R(t) \geq \phi_L(t)$ when $t \geq \max\set{r_A, r_B}$.
Hence, we assume $t \geq \max\set{r_A, r_B}$ in the proof below.

We show that when $\gs_A^{-1}$ and $\gs_B^{-1}$ satisfy the log-log convexity assumption, $\phi_L(t)$ is maximized at either $x_1 = 0$ or $y_1 = 0$. 
To begin with, note that when $x_1 = 0$,
\begin{equation*}
    \phi_L(t) = \gs_A^{-1}(r_A) \gs_B^{-1}(t/r_A),
\end{equation*}
and when $y_1 = 0$,
\begin{equation*}
    \phi_L(t) = \gs_A^{-1}(t/r_B) \gs_B^{-1}(r_B).
\end{equation*}
Since $t/r_A \geq 1$ and $t/r_B \geq 1$, then the pairs $(r_A, t/r_A)$ and $(t/r_B, r_B)$ are feasible in the optimization, and so
\begin{equation*}
    \phi_R(t) = \max_{\substack{t_A \in [1, r_A],~ t_B \in [1, r_B],\\ t_A t_B = t}} \gs_A^{-1}(t_A)\gs_B^{-1}(t_B).
\end{equation*}
Thus, if $\phi_L(t)$ is indeed maximized at $x_1 = 0$ or $y_1 = 0$, then we have $\phi_L \leq \phi_R$, and $R_C = \phi_R^{-1}$ is a valid rank expansion lower bound.

To that end, let us reuse the notation $f = \gs_A^{-1}$ and $g = \gs_B^{-1}$ for sake of simplicity.
We introduce following functions on $(0, r_A)$ and $(0, r_B)$, respectively:
\begin{equation*}
    \hat{f}(x) = \frac{f(r_A) - f(r_A - x)}{xf'(r_A - x)} 
    \text{~~and~~}
    \hat{g}(x) = \frac{g(r_B) - g(r_B - x)}{xg'(r_B - x)}.
\end{equation*}
Since $r_A - \gs_A(n_A - x)$ is log-log convex and increasing on $[0, n_A]$, its inverse, $p(x) \equiv f(r_A) - f(r_A - x)$, is log-log concave and increasing on $[0, r_A]$. Therefore,
\begin{equation*}
    \frac{d}{dx}\ln p(e^x) 
    = 
    \frac{e^xp'(e^x)}{p(e^x)}
\end{equation*}
is positive and decreasing in $x$. Thus,
\begin{equation*}
    \hat{f}(x) = \frac{p(x)}{xp'(x)}
\end{equation*}
is increasing in $x$. Similarly, $\hat{g}(x)$ is also increasing.

Writing $x = x_1$, we can rewrite $\phi_L$ as 
\begin{align*}
    &\phi_L(t)  \\
    &= 
    \max_{\substack{x_1 \in [0, r_A],~ y_1 \in [0, r_B],\\ (r_A - x_1)(r_B - y_1)\, = \,r_A r_B - t}} 
    n_A n_B - \left(n_A - f(x_1)\right)\left(n_B - g(y_1)\right)\\
    &=
    n_A n_B  
    -\min_{x \in \left[0, r_A - \frac{r_Ar_B - t}{r_B}\right]}
     \left(f(r_A) - f(x)\right)\left(g(r_B) - g\left(r_B - \frac{r_Ar_B - t}{r_A - x}\right)\right).
\end{align*}
In the trivial cases $t = r_Ar_B$, one can directly verify $\phi_R = \phi_L$.
Now assume $t < r_Ar_B$, and thus $x < r_A$.
For simplicity, we write $r_Ar_B - t = c > 0$, $\frac{c}{r_A - x} = s_x \in [\frac{c}{r_A}, r_B]$. It suffices to show
\begin{equation*} 
    \argmin_{x \in [0, r_A - \frac{c}{r_B}]}
    \bigg\{ h(x) := \left(f(r_A) - f(x)\right)\left(g(r_B) - g(r_B - s_x)\right) \bigg\}
    \in 
    \big\{ 0, r_A - \frac{c}{r_B} \big\}.
\end{equation*}
The function $h$ is $C^1$ on $[0, r_A - \frac{c}{r_B}]$, and 
\begin{align*}
    h'(x) 
    &=
    [f(r_A) - f(x)]\cdot \frac{ds_x}{dx}\cdot g'(r_B - s_x) - f'(x) [g(r_B) - g(r_B - s_x)]\\
    &=
    s_x\left[ \frac{f(r_A) - f(x)}{r_A - x}g'(r_B - s_x) - f'(x) \frac{g(r_B) - g(r_B - s_x)}{s_x}\right].
\end{align*}
Since $f$ and $g$ are convex and strictly increasing, $f' > 0$, $g' > 0$ on interval $x \in (0, r_A - \frac{c}{r_B})$,
so on this interval,
\begin{align*} \label{eq:hpprop}
    h'(x) 
    &= 
    s_xf'(x)g'(r_B - s_x) \left(\hat{f}(r_A - x) - \hat{g}\left(\frac{c}{r_A - x}\right)\right)\\
    &\propto_+
    \hat{f}(r_A - x) - \hat{g}\left(\frac{c}{r_A - x}\right).
\end{align*}
By monotonicity of $\hat{f}$ and $\hat{g}$, $h$ can only be increasing, increasing then decreasing, or decreasing on $(0, r_A - \frac{c}{r_B})$.
In any case, $h$ can only attain minimum on the boundary $x \in \set{0, r_A - c/r_B}$.
The proof is then complete.
\end{proof}

\edits{In the example below, we show that when $r_A - \gs_A(n_A - x)$ and $r_B - \gs_A(n_B - x)$ are not log-log convex, i.e., when $f(r_A) - f(r_A - x)$ and $g(r_B) - g(r_B - x)$ are not log-log concave, 
it is possible to have $L_C \geq R_C$.
Thus, Theorem \ref{thm:newlb} cannot always be simplified to Theorem \ref{thm:reduce2R}.}

\begin{example} \label{exp:SkinnyL}
In general, given that $f = \gs_A^{-1}$ is convex, strictly increasing, and $f(0) = 0$, we cannot ignore the $\cL$-shaped grids in maximization. A counterexample is given below.

Consider $r_A = 5$ and
\begin{equation*}
\gs_A^{-1}(x) = f(x) =
\begin{cases}
    x                                             & x \leq 3\\
    \frac{5}{2} + \frac{1}{2}e^{2(x - 3)}         & 3 < x \leq 4\\
    \frac{5}{2} + \frac{1}{2}e^{2} + e^2 (x - 4)  & 4 < x \leq 5
\end{cases}.
\end{equation*}
Then one can check that $f(r_A) - f(r_A - x)$ is not log-log concave and $\hat{f}$ is not monotonically increasing.

We demonstrate that in this case the $L$-shaped grids cannot be discarded. Let us take $g = f$. Consider grid $S = L(1, 1; 5, 5)$ of size 9. We have $\ab{S} \approx 26.17$. 
With a rectangle of area 9 inside the $5 \times 5$ region, the maximum expansion size, rounding up to an integer, is $\lceil \phi_R(9) \rceil = \lceil f(5)\cdot f(9/5) \rceil = 25 < \lfloor\ab{S}\rfloor \leq \lfloor \phi_L(9) \rfloor$. 
Thus in this case we have to return to Theorem \ref{thm:newlb} or Theorem \ref{thm:lb}.  
\end{example}

To see how Theorem \ref{thm:reduce2R} improves the bound in Theorem \ref{thm:lb}, we give the following example.
However, despite that the bound in Theorem \ref{thm:reduce2R} is tighter, it is often no longer concave, and it is in general not possible to apply the bound in Theorem \ref{thm:reduce2R} recursively as in Theorem \ref{thm:nest_lb}.

\begin{example} \label{rmk:compare_bnds}
    First consider $\bC = \bA \otimes \bB$, where $\gs_A(k) = k^{1/2}$ and $\gs_B(k) = k^{1/4}$. If we apply Theorem \ref{thm:lb} or equivalently Theorem \ref{thm:nest_lb}, we find a rank expansion lower bound,
    \begin{equation*}
        \gs_C^{\text{prev}}(k) = k ^{1/4}.
    \end{equation*}
    Now if we turn to Theorem \ref{thm:reduce2R} (or Corollary \ref{cor:bnd_lb} stated after this example), for $k > n_B$, we get a different rank expansion lower bound,
    \begin{equation*}
        \gs_C^{\text{new}}(k) 
        =
         \left(\frac{k}{n_B}\right)^{1/2} \cdot n_B^{1/4} 
        =
         n_B^{-1/4}\cdot k^{1/2},
    \end{equation*}
    and when $k \leq n_B$, $\gs_C^{\text{new}}(k) = \gs_C^{\text{prev}}(k)$. The bound is improved by a factor of $k^{1/4}$. 
    Numerically, with $n_A = n_B = 100$, $k = 10 n_A$, we have $\lceil\gs_C^{\text{prev}}(k)\rceil = 6$, whereas $\gs_C^{\text{new}}(k) = 10$.
    However, note that although $\gs_C^{\text{new}}$ is continuous, it is no longer concave. 
    Hence, it is not possible to apply the new bound recursively.
    
    For another illustration, let us use logarithms for the rank expansion lower bound. 
    Let $\bC = \bA \otimes \bA$ with $\gs_A(k) = \ln(k + 1)$. 
    From Theorem~\ref{thm:nest_lb} (or Theorem \ref{thm:lb}), we get
    \begin{equation*}
        \gs_C^{\text{prev}}(k) = \ln \left(\frac{k}{e-1} + 1\right).
    \end{equation*}
    Using the new bound by Theorem \ref{thm:reduce2R} (or Corollary \ref{cor:bnd_lb}), when $k/n_A \geq e - 1$, we have
    \begin{equation*}
        \gs_C^{\text{new}}(k) 
        =
        \ln(n_A + 1) \cdot \ln(k/n_A + 1).
    \end{equation*}
    Thus in this regime,
    \begin{align}
        \exp\left(\gs_C^{\text{prev}}(k)\right) &= \frac{k}{e-1} + 1 \label{eq:lin_growth};\\
        \exp\left(\gs_C^{\text{new}}(k)\right) &= \left(\frac{k}{n_A} + 1\right)^{\ln(n_A + 1)} \label{eq:exp_growth}.
    \end{align}

    \edits{If we choose $k$ to be proportional to $n_A$, \eqref{eq:exp_growth} eventually surpasses \eqref{eq:lin_growth} as $n_A \rightarrow \infty$, and thus we obtain a tighter lower bound, $\gs_C^{\text{new}}$, on the rank expansion.} Numerically, with $n_A = 100$, $k = 10 n_A$, we have  $\eqref{eq:lin_growth} \approx 583$, $\eqref{eq:exp_growth} \approx 63996$. Thus, $\lceil\gs_C^{\text{prev}}(k)\rceil = 7$, whereas $\lceil\gs_C^{\text{new}}(k)\rceil = 12$.
\end{example}

Finally, recall that we can further simplify the optimization problem in $R_C$ using Lemma \ref{lem:bndoptim} when $\gs_A$ and $\gs_B$ are log-log concave.
By combining Theorem \ref{thm:reduce2R} with Lemma \ref{lem:bndoptim}, we have an even simpler expression of the rank expansion lower bound.

\begin{corollary} \label{cor:bnd_lb}
    Let everything be defined as in Theorem \ref{thm:newlb}, if $\gs_A(x)$, $\gs_B(x)$ are log-log concave, $r_A - \gs_A(n_A - x)$ and $r_B - \gs_B(n_B - x)$ are log-log convex, on $(0, n_A)$ and $(0, n_B)$ respectively, e.g., polynomials and logarithm functions listed in Proposition \ref{prop:monom_log},
    then
    \begin{equation*}
        R_C(k) = \min_{\substack{k_A\in [d_A, n_A],~k_B \in [d_B, n_B],\\
        k_A\in\set{d_A, n_A} \text{ or } k_B \in \set{d_B, n_B},\\
        k_A k_B \geq k}}
          \sigma_A(k_A) \cdot \sigma_B(k_B)
    \end{equation*}
    is a rank expansion lower bound for $\bC$.
\end{corollary}




%
%

\end{document}